\documentclass[iopjournal,10pt,onecolumn,superscriptaddress,floatfix,nofootinbib,showpacs,longbibliography]{revtex4-2}
\usepackage{makecell}
\newcolumntype{C}[1]{>{\centering\arraybackslash}m{#1}}
\usepackage{booktabs,adjustbox}
\usepackage{float}
\usepackage[normalem]{ulem}
\usepackage[utf8]{inputenc}  
\usepackage[T1]{fontenc}     
\usepackage[table]{xcolor}
\usepackage{color,soul}
\usepackage[british]{babel}  
\usepackage[scaled=0.86]{berasans}  
\usepackage[colorlinks=true, citecolor=blue, urlcolor=blue]{hyperref}  
\usepackage{graphicx} 
\usepackage[babel]{microtype}  
\usepackage{amsmath,amssymb,amsthm,bm,amsfonts,mathrsfs,bbm} 

\usepackage{xspace}  
\usepackage{pgf,tikz}
\usepackage{xcolor}
\usepackage{multirow}
\usepackage{array}
\usepackage{bigstrut}
\usepackage{braket}
\usepackage{color}
\usepackage{natbib}
\usepackage{multirow}
\usepackage{mathtools}
\usepackage{float}
\usepackage{blkarray}
\usepackage[caption = false]{subfig}
\usepackage{xcolor,colortbl}
\usepackage{color}
\usepackage{rotating}
\usepackage{tikz}
\usepackage{tabularx}
\usepackage{cancel}

\newcommand{\Tr}{\operatorname{Tr}}

\newcommand{\be}{\begin{equation}}
\newcommand{\ee}{\end{equation}}
\newcommand{\ba}{\begin{eqnarray}}
\newcommand{\ea}{\end{eqnarray}}
\newcommand{\ketbra}[2]{|#1\rangle \langle #2|}

\newtheorem{theorem}{Theorem}

\newtheorem{definition}{Definition}
\newtheorem{proposition}{Proposition}

\newtheorem{example}{Example}

\newtheorem{lemma}{Lemma}

\def\>{\rangle}
\def\<{\langle}

\usepackage{centernot}
\usepackage{subfig}
\usepackage[dvipsnames]{xcolor}

\begin{document}

\title{Demultiplexing Generalized Information via Quantum Transmission Lines} 
\author{Soham Sau}
\affiliation{RCQI, Institute of Physics, Slovak Academy of Sciences, Dúbravská cesta 9, 84511 Bratislava, Slovakia}
\author{Anna Jenčová}
\affiliation{Mathematical Institute, Slovak Academy of Sciences,  Štefánikova 49, 814 73 Bratislava, Slovakia}
\author{Tamal Guha}
\affiliation{Mathematical Institute, Slovak Academy of Sciences,  Štefánikova 49, 814 73 Bratislava, Slovakia}
\begin{abstract}
   Demultiplexers are the fundamental primitives of network architecture, enabling perfect routing of an input classical signal to a designated one, among multiple output ports. Quantum transmission lines, having access to the quantum systems directly, are able to transmit both the classical and quantum information encoded in quantum systems. A natural question therefore emerges that whether the scrambled classical and quantum information in a quantum system can be perfectly demultiplexed in the designated classical and quantum output ports? Here we answer this question by introducing a quantum to quantum-classical device, namely the quantum demultiplexer (Q-DEMUX). We characterize the class of Q-DEMUXs enabling perfect routing of both the classical and the quantum information along with their simple circuit realizations. Our results highlight an explicit connection between the strength of a Q-DEMUX with the incompatibility of quantum instruments. Finally, we extend the notion in a stronger variant where the sender is oblivious regarding the nature of the data to be transmitted through the Q-DEMUX.
\end{abstract}
\maketitle
\section{Introduction}
In a classical network, it is often required to route an input classical signal to one of the assigned output nodes, leaking no information to the others. From the simple telephone system to the modern days Packet switching, Multi-Protocol Label switching configurations for internet distributions are some of the instances, appeared in the telecommunication networks. A prototypical device to employ such a data-routing process is the demultiplexer (DEMUX) \cite{mano1972computer}. A 1-to-N DEMUX is essentially a Boolean circuit with one input, N output ports and a selector. By selecting the designated output port through the selector, the sender can route the input data to the corresponding output port.

With the advent of quantum information theory, along with the rapid development of quantum networks \cite{komar2014quantum, wehner2018quantum, pant2019routing} one may tempted to ask about demultiplexing the quantum information among multiple quantum nodes. While configuring the explicit quantum demultiplexer in an optical set-up was attempted in \cite{garcia2009quantum, xie2018quantum}, the implicit notion for the same has also been reported in \cite{bhattacharya2021random, guha2023quantum} with the coherent control over noisy quantum transmission lines. However, the quantum transmission lines, a.k.a, quantum channels, being capable of transmitting quantum systems directly could communicate both the classical as well as quantum information in any quantum network. Keeping this in mind, here we ask how to demultiplex the generalized information encoded in a quantum system in any network. More precisely, when a quantum system is transmitted through a channel, is it possible to send the encoded information to a quantum or a classical output port, depending upon the nature of data? This suggests for a quantum device with a single input quantum port, along with a quantum and a classical output port to store the encoded information in the input system. In addition, the sender (say, Alice) is equipped with a binary selector, to select the routing port depending upon the nature of the encoded information. Mimicking the structure of classical demultiplexer, the generalized information demultiplexer, namely the Q-DEMUX, sends a junk information to the classical port whenever the encoded information is quantum in nature and vice-versa. 

Here, we first show that such a Q-DEMUX can always be identified in terms of a pair of quantum instruments realizing the effective quantum channel between the input and output quantum port of the Q-DEMUX. The quantum instruments are the generalized quantum device, encapsulating both the quantum channels and the quantum measurement in a single set-up \cite{heinosaari2011mathematical,busch2016quantum,davies1970operational}. Accordingly, the demultiplexing strength of a Q-DEMUX is quantified in terms of the classical as well as quantum information processing of the quantum transmission line, in terms of the given instrument realizations. Setting an upper-bound for the demultiplexing strength of a general Q-DEMUX, we then characterize the class of quantum channels, for which there exists a Q-DEMUX configuration allowing to transfer both the classical and the quantum information perfectly in the respective output ports. Interestingly, the set of such channels are identified as a subset of entanglement breaking channels \cite{horodecki2003entanglement}, for which traditionally no quantum information can be communicated between two distant parties. For classical information processing tasks, however, the power of entanglement breaking channels is well studied \cite{shirokov2012conditions, chiribella2025communication, vieira2025entanglement}. While the circuit realization of a perfect Q-DEMUX may require an high dimensional ancillary system, we provide a simple circuit model to implement a specific class of perfect Q-DEMUXs. In particular, if for a perfect Q-DEMUX there exists a class of input quantum states for which the output states are also pure states of the same input Hilbert space, then such a Q-DEMUX can be implemented with ancillary quantum system of dimension identical to that of the input. Additionally, for any Q-DEMUX we set an upper-bound on the demultiplexing strength, the violation of which necessarily certifies the presence of incompatibility between the pair of instrument realizations of a quantum channels resulting the best possible classical and the quantum information demultiplexing.

Furthermore, we consider a stronger variant of Q-DEMUX, for which the sender is completely oblivious regarding the nature of data to be transmitted and hence do not require any access to the selector. In other words, the proper receiver, that is, whether the encoded information should be stored in the classical or the quantum output port is not known to Alice. Such an instance of random receiver quantum information transmission is particularly interesting in the contexts of delegated quantum computation models \cite{childs2005secure, arrighi2006blind, fitzsimons2017private}. In the present scenario, analogously the selector-less demultiplexing strength replicate the similar picture in question of delegated quantum-classical communication model. While similar to the earlier case here we derive an upper-bound for the selector-less demultiplexing strength, we only characterize two different classes of quantum channels for which there exists a Q-DEMUX configuration achieving the optimal selector-less demultiplexing strength. We also conjecture that these are the only possible quantum channels, for which there exists a single Q-DEMUX configuration to obtain the optimal selector-less demultiplexing strength. Finally, we conclude with the possibilities of extending the notion of selector-less Q-DEMUX configuration towards estimating the general information processing abilities of a quantum instrument.  
%
%
\section{Mathematical Framework}
Classically, a demultiplexer is a device with a single classical input port with an \(n\)-bit string \(X^n\) and with \(m\) possible output ports, each capable of storing an \(n\)-bit classical message. The sender having an access to an additional \(\lceil\log_2m\rceil\)-bit strings, namely the selectors, would able to transfer the input classical message \(X^n\) in any of the desired output ports among \(\{1,2,\cdots,m\}\). 

In a similar fashion, it is possible to design such demultiplexers capable of transmitting the quantum information. Advancing in this direction, here we propose a quantum demultiplexer (Q-DEMUX), which is capable of separating the classical or the quantum information stored in the input quantum system. This implies a Q-DEMUX \(\mathrm{D}\) can be operationally seen as a device with a single input quantum port (\(A\)) along with two kinds of output ports: the first one (\(Q\)) stores the quantum information, while the second one (\(C\)) is a classical register only. The sender (say Alice), by tuning a binary selector \(s\in\{0,1\}\) is able to transfer the information in the designated port depending upon it's nature. Mathematically, such a device can be represented as a quantum to quantum-classical completely positive trace preserving (CPTP) map 
\[\mathrm{D}_s:\mathcal{D}(\mathcal{H}_{in})_A\to\mathcal{D}(\mathcal{H}_{out})_Q\times X_C,\]
where, \(\mathcal{D}(\mathcal{H}_{in})\) and \(\mathcal{D}(\mathcal{H}_{out})\) denotes the set of density operators on the input and the output Hilbert spaces, respectively. Additionally, \(X\) is the set of classical random variables \(x\in \{0,\cdots,|X|-1\}\) stored in the classical register \(C\). Incorporating the role of the selector, we can further represent the action of the Q-DEMUX \(\mathrm{D}_s\), such that, for every \(\rho\in\mathcal{D}(\mathcal{H}_{in})\),
\begin{align}\label{e1}
\mathrm{D}_s(\rho_A)=&\sum_xp(x,\rho,s)\sigma(x,\rho,s)_Q\otimes\ketbra{x}{x}_C\\\nonumber&=\bigoplus_{x}p(x,\rho,s)\sigma(x,\rho,s)_Q.
\end{align}\noindent
Here, \(p(x,\rho,s)\) denotes the probability at which the classical register \(C\) stores \(x\in X\) for the input state \(\rho\) and the selector choice \(s\); and the corresponding output state is designated as \(\sigma(x,\rho,s)\) at the output quantum port \(Q\). From now on, we often denote the classical indices \(\{x\in X\}\) as the computational basis \(\{\ketbra{x}{x}\}_x\) of a \(|X|\)-dimensional quantum system. Additionally, since the states in the classical register \(\{\ketbra{x}{x}\}_x\) are mutually orthogonal, we can also represent the action as a direct sum of all the quantum states \(\sigma(x,\rho,s)\) appeared at \(Q\) with a probability of \(p(x,\rho,s)\) as shown in Eq.(\ref{e1}). Further, the suffixes \(\{A,Q,C\}\) there respectively denotes the Alice's input, quantum output and the classical output ports, which we will drop occasionally in the rest of the manuscript when there is no confusion. 

Also note that, we append the selector \(s\in\{0,1\}\) to the Q-DEMUX \(\mathrm{D}_s\), which can in general be considered as a two-level quantum system \(\ket{s}\in\{\ket{0},\ket{1}\}\). Therefore, we can also denote the action of the Q-DEMUX by invoking the role of the selector qubit, as 
\begin{align*}
    \mathrm{D}_s:\mathcal{D}(\mathcal{H}_{in})_A\otimes\mathbb{C}_s^2\to\mathcal{D}(\mathcal{H}_{out})_Q\times X_C.
\end{align*}
 Here, the two-dimensional complex Hilbert space \(\mathbb{C}^2\) corresponds to the selector qubit.

Within this framework, we can now define the classical information processed to the classical register \(C\). In general, the output classical random variable can be characterized as a new random variable \(z\in Z\), extracted from the output \(x\in X\). That is, we can define an arbitrary \(|X|\)-to-\(|Z|\) function \(f\), such that, the true output random variable \(z=f(x)\). Now, to process the classical information, the input random variable \(y\in Y\) is encoded to a set of quantum states \(\{\rho_y\in\mathcal{D}(\mathcal{H}_{in})\}_y\) and sent it through the Q-DEMUX, along with setting an appropriate value to the selector (say \(s=0\)). Accordingly, the classical demultiplexing strength of the Q-DEMUX with the specific encoding \(\sum_yp_y\rho_y:=\rho\), can be characterized as
\begin{gather}\label{e2}
    \mathcal{C}(\mathrm{D}_0,\rho_{RA})=\max_{f:X\to Z}I_{\rho_{RA}}(Y:Z),
\end{gather}\noindent
 where \(\rho_{RA}=\sum_yp_y\ketbra{y}{y}_R\otimes(\rho_y)_A\) is the classical quantum state between the a reference system \(R\), generating the classical symbols \(y\in Y\)with the probability distribution \(\{p_y\}_y\) and the Alice \(A\), possessing the input port of the Q-DEMUX. Note that, the suffix "0" for \(\mathrm{D}\) denotes the choice of the selector \(s=0\) in context of classical information processing. Moreover, here \(I(Y:Z)=H(p_y)+H(p_z)-H(p_{yz})\) is the mutual information between the input and the processed output random variables (\(Y\) and \(Z\) respectively) and \(H(p_m)=-\sum_mp_m\log_2p_m\) denotes the Shannon entropy for the random variable \(m\in M\) following the probability distribution \(\{p_m\}_m\). In particular, here the joint probability \(p_{yz}=\sum_xf(z|x)p(x,\rho_y,s=0)\) and \(p_z=\sum_yp_{yz}\).
 
 On the other hand, to process the quantum information, one could apply a predefined CPTP map on the output quantum state, depending upon the classical outcome registered in \(C\). In general, let us consider a \(|X|\)-to-\(|T|\) function \(g\), such that \(T\ni t=g(x)\) and accordingly the set of CPTP maps \(\{\mathcal{N}_t\}_t\) can be applied on the output quantum state at \(Q\). 
Let us now consider a bipartite quantum state \(\ket{\psi}_{RA}\in\mathcal{H}_{in}^{\otimes 2}\) between the reference (\(R\)) and Alice (\(A\)), and the subsystem \(\rho_A=\Tr_R(\ketbra{\psi}{\psi})\) is sent through \(\mathrm{D}_s\) with selector choice \(s=1\).
 Now for each of the output classical random variable \(x\in X\), a correcting operation \(\mathcal{N}_{t=g(x)}:\mathcal{D}(\mathcal{H}_Q)\to \mathcal{D}(\mathcal{H}_B)\) is applied on the output quantum state, where \(\mathcal{H}_Q\cong\mathcal{H}_{out}\). Then the effective quantum demultiplexing strength of \(\mathrm{D}_1\) with the specific \(\rho_A\in\mathcal{D}(\mathcal{H}_{in})\), optimizing over all possible correcting operations \(\{\mathcal{N}_t\}_t\), can be defined as
\begin{align}\label{e3}
     \mathcal{Q}(\mathrm{D}_1,\ket{\psi}_{RA})=\max_{g:X\to T}\max_{\{\mathcal{N}_t\}_t}  I_{\sigma(\psi,s=1)}(R\rangle B),
 \end{align}\noindent
 where, \(I_{\sigma}(R\rangle B)=\max\{S(\sigma_B)-S(\sigma_{RB}),0\}\) is the coherent information of the quantum state \(\sigma_{RB}\), with \(S(\sigma)=-\Tr(\sigma \log\sigma)\) is the von-Neumann entropy of any quantum state \(\sigma\). Here, the effective corrected state  
 \begin{align}\label{e3.5}
 \sigma(\psi,s=1)_{RB}=[id_R\otimes \sum_{x,t}p(x,\rho_A,s=1)g(t|x)(\mathcal{N}_t)_{Q\to B}]\sigma(x,\rho_A,s=1)_{RQ}.\end{align}
 From now on, by \(id_\alpha\) we denote the identity map on \(\mathcal{L}(\mathcal{H}_\alpha)\). Moreover, \(\sigma_{RQ}(x,\psi,s=1)\) is the uncorrected output quantum state between \(R\) and the quantum register \(Q\), when the classical register clicks \(x\) and the selector is set to be \(s=1\). The non-signaling condition now implies 
 \[\rho_R=\Tr_A(\ketbra{\psi}{\psi}_{RA})=\Tr_B[\sigma(\psi,s=1)_{RB}=\sum_xp(x,\rho_A,s=1)\Tr_Q[\sigma(x,\psi,s=1)_{RQ}]].\]
 It is important to mention that the Eq. (\ref{e3}), in principle, should involve another maximization over all possible purification \(\ket{\phi}_{RA}\in\mathcal{H}_R\otimes\mathcal{H}_{in}\) of input quantum state \(\rho_A=\Tr_R(\ketbra{\phi}{\phi})\). However, all of them are connected with \(\ket{\psi}_{RA}\in\mathcal{H}_{in}^{\otimes 2}\) with \(\rho_A\) as the marginal, under local isometry on the subsystem \(R\) and since, local isometries can not change the coherent information, we have avoided the maximization. Furthermore, for the classical part in Eq. (\ref{e2}), the classical-quantum state \(\rho_{RA}\) for the input quantum state \(\rho_A=\Tr_R(\rho_{RA})\) is unique up to the local unitary on \(R\). Therefore, we have dropped the maximization there too.
 
 \begin{figure}[t]
    \centering
    \includegraphics[width=0.65\linewidth]{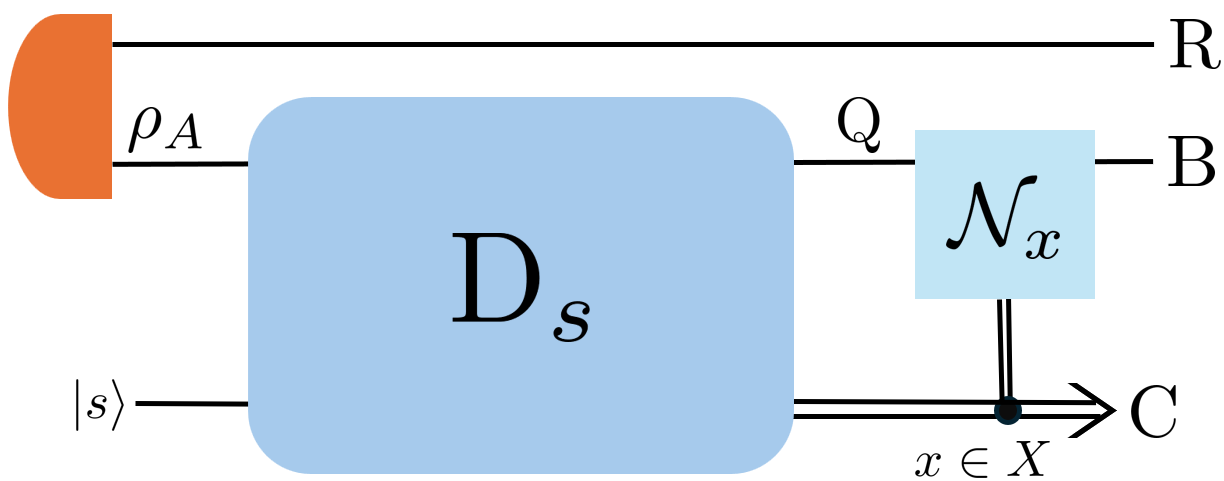}
    \caption{\textit{Schematic diagram of a Q-DEMUX.} A bipartite state between Alice \(A\) and the Reference \(R\) is prepared and Alice's marginal state \(\rho_A\) is sent through the Q-DEMUX. Depending upon the nature of information encoded in \(\rho_A\), the selector \(s\in\{0,1\}\) is chosen. Finally, at the output end depending upon the classical random variable \(x\in X\) registered in \(C\) a pre-decided correction operation \(\mathcal{N}_x\) is applied on the output quantum state at \(Q\).}
    \label{fig1}
\end{figure} 

Now, in the context of demultiplexer one can visualize both the classical and the quantum information processing in a single picture. The same marginal quantum state \(\rho_A\in\mathcal{D}(\mathcal{H}_{in})\) sent through the Q-DEMUX can share a specific correlation with the reference system depending upon the nature of data encoded in it. In particular, the marginal state shares a classical correlation with the reference system \(R\) in question of classical information processing, while the same marginal can be purified in the form of an entangled state with \(R\) to process the quantum information. In other words, operationally a Q-DEMUX takes an input quantum state at its input port and depending upon the information encoded in it Alice can route the system in a specific way, such that the encoded information stores in the corresponding output port. Within this framework we can now, define the total demultiplexing strength of a Q-DEMUX \(\mathrm{D}_s\) as
\begin{align}\label{eq:e4}
     \mathcal{T}_s(\mathrm{D}_s)=\max_{\rho_A\in\mathcal{D}(\mathcal{H}_{in})}[\mathcal{C}(\mathrm{D}_0,\rho_{RA})+\mathcal{Q}(\mathrm{D}_1,\ket{\psi}_{RA})],\text{ with }\rho_A=\Tr_R(\rho_{RA})=\Tr_R(\ketbra{\psi}{\psi}_{RA})
 \end{align}
 where, \(\mathcal{C}(\mathrm{D}_0,\rho_{RA})\) and \(\mathcal{Q}(\mathrm{D}_1,\ket{\psi}_{RA})\) are same as in Eqs. (\ref{e2}) and (\ref{e3}). Also, note that due to the uniqueness of \(\rho_{RA}\) and \(\ket{\psi}_{RA}\) for a given \(\rho_A\in\mathcal{D}(\mathcal{H}_{in})\) it is sufficient to optimize over the input quantum states in \(\mathcal{D}(\mathcal{H}_{in})\) in Eq. (\ref{eq:e4}). Operationally, these quantity suggests the maximal strength of Q-DEMUX to demultiplex the generalized information encoded in a single quantum state. We refer to Fig. \ref{fig1} for the schematic idea of the information demultiplexing process.  
 
  It is now important to emphasize on the role of the selector to configure the Q-DEMUX depending upon the nature of the generalized information the sender is supposed to process. For instance, consider a Q-DEMUX, where the input quantum state without any error transferred to the output port \(Q\), along with the computational basis measurement performed on an arbitrary ancillary quantum state to register the classical outcome at \(C\). This configuration allows perfect transmission of the quantum information, even without any further correction. Now, if the selector is so powerful that it could just swap the input and the ancillary qudit during the transmission of the classical information, then again the classical information can be perfectly demultiplexed through it, by encoding the classical random variable in the computation basis. Although it demultiplexes the information perfectly, this trivial construction admits no further restriction on the structure of the Q-DEMUX. While to make it more operationally feasible one may come up with various possible constraints on the architecture of the Q-DEMUX, here we consider a fairly simple one: \emph{The quantum channel induced from the Q-DEMUX must be invariant under the choice of the selector}. Using Eq.(\ref{e1}), this condition mathematically reads
 \[\Tr_C[\mathrm{D}_s(\rho_A)]=\sum_xp(x,\rho,s)\sigma(x,\rho,s)_Q=\sigma(\rho)_Q,\]
 for every possible input quantum state \(\rho_A\in\mathcal{D}(\mathcal{H}_{in})\) and for every choice of the selector \(s\in\{0,1\}\). Since the action in Eq. (\ref{e1}) is a CPTP map, we can identify the above operation as a CPTP-map \(\Lambda_{\text{ind}}:\mathcal{D}(\mathcal{H}_{in})\to\mathcal{D}(\mathcal{H}_{out})\), namely the induced quantum channel.
 %
 %
 %
 %
 \section{Main Results}
 In the abstract sense, a Q-DEMUX \(\mathrm{D}_s\) can be thought of a black box consists of two different quantum-to-quantum-classical devices \(\mathrm{D}_0\) and \(\mathrm{D}_1\). Moreover, the effective quantum channel should be identical for these two devices. Depending upon the nature of the data encoded in the input quantum state \(\rho_A\), one of these two devices are chosen by the binary selector \(s\). In this section, we will characterize the physical realizations of such Q-DEMUXs and their utilities in question of routing the input information to the designated output ports.
 \subsection{The Structure of Q-DEMUXs}
As discussed earlier, the output quantum state (at the port \(Q\)) and classical random variable (at \(C\)) from the Q-DEMUX may required to process further to obtain the encoded classical or quantum information. For a given Q-DEMUX \(\mathrm{D}_s\), one should therefore optimize over the triplet \((f,g,\{\mathcal{N}_x\}_x)\) to achieving the maximum demultiplexing strength. While intuitive, in the following we will first show that the functions \(f:X\to Z\) and \(g:X\to T\) are actually redundant. Therefore, to maximize the demultiplexing strength ,the Q-DEMUX is only appended with a set of correcting operations optimized over the CPTP maps from \(\mathcal{D}(\mathcal{H}_{Q})\) to \(\mathcal{D}(\mathcal{H}_B)\).
 \begin{lemma}\label{l1}
     To maximize \(\mathcal{T}_s(\mathrm{D}_s)\) of \(\mathrm{D}_s\), both the functions \(f:X\to Z\) and \(g:X\to T\) can be taken to be identity.    
 \end{lemma}
 \begin{proof}
     The proof follows directly from the data-processing inequality. 
     
     Since, the random variable \(Z\) is post-processed from the output random variable \(X\) in the classical register \(C\), 
     for any input quantum ensemble, encoded with the classical random variable \(Y\),
     \[I(Y:Z)\leq I(Y:X).\]

     For the quantum information part, consider the corrected average quantum state as in Eq.(\ref{e3.5}), which can be rewritten as 
     \[\sigma(\psi,s=1)_{RB}=[id_R\otimes \sum_{x}p(x,\rho_A,s=1)\sum_tg(t|x)\mathcal{N}_t]\sigma(x,\rho_A,s=1)_{RQ}.\]
     Now, by replacing \(\sum_tg(t|x)\mathcal{N}_t=\mathcal{N}_x,\) as a new set of correction operations for all \(x\in X\), we can omit the random variable \(t=g(x)\).
 \end{proof}
While the set of correcting operations \(\{\mathcal{N}_x\}_x\) are externally appended to the Q-DEMUX, the inherent configuration of it depends upon the induced quantum channel \(\Lambda_{\text{ind}}\) and the classical random variables \(X\). Therefore, one can ask given a \(\Lambda_{\text{ind}}\), how to identify the set \(X\), i.e., the Q-DEMUX \(\mathrm{D}_s\) to maximize \(\mathcal{T}_s(\mathrm{D}_s)\). In the following we will answer this in terms of quantum instruments \cite{davies1970operational, busch2016quantum, heinosaari2011mathematical}.

Let us first recall that a quantum instrument $\mathcal D(\mathcal H_{in})\to \mathcal D(\mathcal H_{out})$ with outcomes in a (finite) set $X$ is a collection of completely positive trace non-increasing (CPTNI) maps $\mathcal{I}_x:\mathcal{D}(\mathcal{H}_{in})\to \mathcal{D}(\mathcal{H}_{out})$ such that 
$\mathcal{I}_{ind}:=\sum_x\mathcal{I}_x$ is trace preserving. Alternatively, an instrument can be
viewed as a channel with both quantum and classical outcomes, that is
\[
\mathcal{I}(\rho)=\sum_x \mathcal{I}_x(\rho)\otimes |x\rangle\langle x|.
\]
The two marginals of $\mathcal{I}$ are the induced channel $
\mathcal{I}_{ind}$ and the induced measurement 
\[
\rho\mapsto \sum_x
\Tr[\mathcal{I}_x(\rho)]|x\>\<x|,
\]
with the induced POVM $\{A_x^{\mathcal{I}}:=\mathcal{I}_x^*(I_{out})\}_{x\in X}$, determined by the condition 
$\Tr[A_x^\Lambda\rho]=\Tr
[\mathcal{I}_x(\rho)]$, for all $\rho\in \mathcal{D}(\mathcal{H}_{in})$. The instrument $\mathcal{I}$ is then called an instrument realization of the channel $\mathcal I_{ind}$.

An isometric dilation
of an instrument $\mathcal{I}$ is a triple $(\mathcal{H}_E, V, M)$, where $\mathcal{H}_E$ is a
Hilbert space, $V: \mathcal{H}_{in}\to \mathcal{H}_{out}\otimes \mathcal{H}_E$ is an isometry
and $M=\{M_x\}_{x\in X}$ is a POVM on $\mathcal{H}_E$, such that 
\[
\mathcal{I}_x(\rho)=\Tr_E[V\rho V^\dagger(I_{out}\otimes M_x)],\qquad \rho\in
\mathcal{D}(\mathcal{H}_{in}),\ x\in X.
\]
For the induced channel and POVM, we have
\[
\mathcal{I}_{ind}(\rho)=\Tr_E[V\rho V^\dagger],\ \rho\in \mathcal{D}(\mathcal{H}_{in});
\qquad A^\mathcal{I}_x=V^\dagger(I_{out}\otimes
M_x)V,\ x\in X.
\]
In particular, $(\mathcal{H}_E,V)$ is a Stinespring dilation of the channel $\mathcal{I}_{ind}$. Recall
that a Stinespring dilation is minimal if 
\[
\mathrm{span}\{(K\otimes I_E)V|\xi\>,\  K\in \mathcal{L}(\mathcal{H}_{out}),\ |\xi\>\in
\mathcal{H}_{in}\}=\mathcal{H}_{out}\otimes \mathcal{H}_E.
\]
In this case, the ancillary POVM $M$ is given uniquely. Moreover, any other Stinespring
dilation $(\mathcal{H}_{E'},V')$ of $\mathcal{I}_{ind}$ is connected to a minimal dilation
$(\mathcal{H}_E,V)$ by an isometry
$U:\mathcal{H}_E\to\mathcal{H}_{E'}$ such that $V'=(I_{out}\otimes U)V$,  and
$(\mathcal{H}_{E'},V',M')$ with $M':= UM U^\dagger$,  is an isometric dilation of the
instrument $\mathcal{I}$. 
Since any isometry can be extended to a unitary, we see that any pair of isometric
dilations $(\mathcal{H}_{E_1}, V_1, M_1)$ and $(\mathcal{H}_{E_2}, V_2, M_2)$ of
$\mathcal{I}$
are connected by a unitary $U:\mathcal{H}_{E_1}\to \mathcal{H}_{E_2}$ (enlarging one of the
spaces if necessary), such that 
\[
V_2=(I_{out}\otimes U)V_1,\qquad M_2=UM_1 U^\dagger. 
\]
\begin{lemma}\label{l2}
    Any Q-DEMUX \(\mathrm{D}_s\) can be configured in terms of two instrument realizations of the induced quantum channel \(\Lambda_{\text{ind}}\).
\end{lemma}
\begin{proof}
Note that, any Q-DEMUX for each of the classical outcome \(x\in X\) in the classical register \(C\) corresponds to a stochastic transformation \(\mathcal{D}(\mathcal{H}_{in})\to\mathcal{L}(\mathcal{H}_{out})\) given by
 \begin{align}\label{e5.5}
   \rho_A\mapsto p(x,\rho,s)\sigma_Q(x,\rho,s),~\forall x\in X.
   \end{align}
   Since, these transformations are physically realizable by observing the classical random variable \(x\in X\) of the Q-DEMUX \(\mathrm{D}_s\), they must correspond to the physically realizable stochastic quantum operations. Therefore, each of these transformations must be linear and completely positive in nature. Moreover, by ignoring the classical index \(x\in X\) we obtain 
   \[\rho_A\mapsto\sum_xp(x,\rho,s)\sigma_Q(x,\rho,s)=:\Lambda_{\text{ind}}(\rho_A),\]
   as the effective induced quantum channel \(\Lambda_{\text{ind}}:\mathcal{D}(\mathcal{H}_{in})\to\mathcal{D}(\mathcal{H}_{out})\), independent of the choice of \(s\).

  Therefore, it follows that the maps  in \eqref{e5.5} define a pair of  instrument realizations 
$\{\Lambda_{s,x}\}_{x\in X}$, $s=0,1$, of the channel $\Lambda_{\text{ind}}$,
such that
\[
\mathrm{D}_s(\rho)=\sum_x \Lambda_{s,x}(\rho)\otimes |x\>\<x|,\qquad s=0,1.
\]
We next show that incorporating the selector as an input, $\mathrm{D}$ can be realized as an instrument
$\mathcal{D}(\mathcal{H}_{in}\otimes \mathbb{C}^2)\to \mathcal D(\mathcal{H}_{out})$, such
that $\mathrm{D}(\rho\otimes |s\>\<s|)=\mathrm{D}_s(\rho)$. Indeed, we can construct a minimal dilation of
such an instrument $\mathrm{D}$ as follows. Let
$(\mathcal{H}_E,V)$ be a minimal Stinespring dilation for $\Lambda_{\text{ind}}$, then the two
instrument realizations $\{\Lambda_{s,x}\}$ have minimal dilations of the form  $(\mathcal{H}_E,V,M_s)$, for
some POVMs $\{M_{s,x}\}_x$, $s=0,1$, on $\mathcal{H}_E$. Put
\[
\mathcal{H}_{E_s}:=\mathcal{H}_E\otimes \mathbb C^2,\qquad  
V_s:=V\otimes I_2, \qquad M_{s,x}:=M_{0,x}\otimes|0\>\<0|+M_{1,x}\otimes|1\>\<1|
\]
and define $\mathrm{D}$ as the instrument corresponding to $(\mathcal{H}_{E_s},V_s,M_s)$. It is then easy to see that the
dilation is minimal and 
\[
\mathrm{D}(\rho\otimes |s\>\<s|)=\sum_x\Tr_{E_s}[V_s(\rho\otimes |s\>\<s|)V^\dagger(I_{out}\otimes
M_{s,x})]\otimes |x\>\<x|=\sum_x \Lambda_{s,x}(\rho)\otimes |x\>\<x|.
\]
  
This implies that the Q-DEMUX \(\mathrm{D}\), with an access to the selector qubit \(|s\>\in\{|0\>,|1\>\}\) can be configured in terms of two of the instrument realizations of the induced quantum channel \(\Lambda_{\text{ind}}\).
\end{proof}

Using the above construction, we can relate a Q-DEMUX to a pair of instrument realizations
of a given channel \(\Lambda_{\text{ind}}\) (adding some zero operators to one of the ancillary POVMs
if the number of outcomes do not match). This will be called a Q-DEMUX realization of the
channel \(\Lambda_{\text{ind}}\).

\subsection{Maximizing the Demultiplexing Strength}
With the sufficiency of identifying the Q-DEMUXs as instrument realizations of quantum
channels for optimal demultiplexing strength, we can now rephrase our question in terms of
the quantum channel induced from a Q-DEMUX. In particular, we will now characterize the
form of the quantum channels for which there exists Q-DEMUX realizations achieving the
maximum value of total demultiplexing strength. We begin with a trivial upper-bound on \(\mathcal{T}_s(\mathrm{D}_s)\).
\begin{lemma}\label{l3}
    For any Q-DEMUX, \(\mathrm{D}_s:\mathcal{D}(\mathcal{H}_{in})\to\mathcal{D}(\mathcal{H}_{out})\times X\), with an access to the selector, the total demultiplexing strength satisfies \(\mathcal{T}_s(\mathrm{D}_s)\leq 2\log_2 d_{in}\).
\end{lemma}
\begin{proof}
    The proof follows from the fact that even with the independent maximization each of the quantum or classical capacities are bounded by \(\log_2 d_{in}\). Precisely,
    \begin{align*}
        \mathcal{T}_s(\mathrm{D}_s)\leq \max_{\sum_yp_y\rho_y\in\mathcal{D}(\mathcal{H}_{in})}\mathcal{C}(\mathrm{D}_0,\rho_{RA})+\max_{\ket{\psi}\in\mathcal{H}_{in}^{\otimes 2}}\mathcal{Q}(\mathrm{D}_1,\ket{\psi}_{RA})\leq 2\log_2 d_{in},
    \end{align*}
    where, \(\sum_yp_y\rho_y=\Tr_R[\rho_{RA}]\) need not to be identical with the marginal \(\Tr_R[\ketbra{\psi}{\psi}]\), in general.
    \end{proof}
    \noindent
Note that while the total demultiplexing strength is defined via maximizing the sum of the classical and quantum capacity,  the upper-limit derived above involves independent maximization of individual capacities. Therefore, it is a natural question whether the upper-bound is practically achievable. In this section, we will show that this is indeed possible. In Theorem \ref{t1} below, we characterize the class of quantum channels that have a  Q-DEMUX realization that achieves the bound. The following propositions are instrumental for this result.

\begin{proposition}\label{p1}
A quantum channel \(\Lambda_{\text{ind}}\) has a Q-DEMUX  realization \(\mathrm{D}_1\) with 
\(\mathcal{Q}(\mathrm{D}_1)=\log_2 d_{\text{in}}\) if and only if \(\Lambda\) is a convex combination of isometries (RI) \(\{\mathcal{V}_k:\mathcal{H}_{in}\to\mathcal{H}_{out}\}_k\).
    \\\noindent
    This also implies that for \(\mathcal{H}_{in}\simeq\mathcal{H}_{out}\), \(\Lambda\) is a random-unitary channel.
\end{proposition}
\begin{proof}
 To prove the \textit{if} part, let \(\Lambda\) be a random isometric channel \[\Lambda(\rho)=\sum_kp_k\mathcal{V}_k\rho\mathcal{V}_k^{\dagger},~\forall\rho\in\mathcal{D}(\mathcal{H}_{in}),\]
where, each of the \(\mathcal{V}_k:\mathcal{H}_{in}\to\mathcal{H}_{out}\) is an isometry appearing with  probability \(p_k\). Consider the instrument realization \(\mathcal{I}(\Omega=[k],\mathcal{H}_{in},\mathcal{H}_{out})\) of \(\Lambda\) such that 
\begin{equation}\label{eq:instrument_RI}
\mathcal{I}_x(\rho)=p_x\mathcal{V}_x\rho\mathcal{V}_x^\dagger,\quad    \forall x\in[k].
\end{equation}
Then we can construct a Q-DEMUX \(\mathrm{D}_1\) such that with the selector choice \(s=1\),
the classical register stores the random variable \(x\in[k]\) and the corresponding output state at the quantum port will be \(\mathcal{V}_x\rho\mathcal{V}_x^{\dagger}\). Additionally, let \(P_x\) be the projections on \(\mathcal{H}_{\text{out}}\) given by
\(P_x:=\mathcal{V}_x\mathcal{V}_x^\dagger\) and define the correcting operations as
\[
\mathcal{N}_x(\sigma)=\mathcal{V}_x^{\dagger}\sigma \mathcal{V}_x+ \Tr[\sigma
(I-P_x)]\omega,
\] 
where $\omega$ is any state on $\mathcal{H}_{\text{out}}$.  Then,  applying the corrections depending upon the classical random variable \(x\) stored in \(C\), we obtain 
\[
\rho\mapsto \sum_x p_x \mathcal{N}_x(\mathcal{V}_x\rho\mathcal{V}_x)=\rho.
\]
 If Alice sends the part of a maximally entangled state
 \(\ket{\phi^+}=\frac1{\sqrt{d_{in}}}\sum_{i=0}^{d_{in}-1}\ket{ii}_{RA}\) through the
 Q-DEMUX, then  \(\mathcal{Q}(\mathrm{D}_1,\ket{\phi^+})=\log_2d_{in}\) attains the optimal value.

To prove the \textit{only if}-part, consider a Q-DEMUX \(\mathrm{D}_1\) with induced
channel \(\Lambda_{\text{ind}}\), such that \(\mathcal{Q}(\mathrm{D}_1)=\log_2d_{\text{in}}\).
Let 
\(\mathcal{I}(\Omega,\mathcal{H}_{in},\mathcal{H}_{out})\) given by the CPTNI maps
\(\{\mathcal{I}_x\}_{x\in\Omega}\), be the instrument corresponding to \(\mathrm{D}_1\). Then there exist an entangled state \(\ket{\psi}_{RA}\in\mathcal{H}_{in}^{\otimes 2}\), and a set of correction operations \(\{\mathcal{N}_x\}_x\) such that 
\[
\mathcal{Q}(\mathrm{D}_1,\ket{\psi}_{RA})=I_{\sigma}(R\rangle Q)=\log_2(d_{\text{in}}),
\text{ where }
\sigma_{RQ}=\sum_x (id_R\otimes
\mathcal{N}_x\circ\mathcal{I}_x)(\ketbra{\psi}{\psi}_{RA}).
\]
By the convexity of the coherent information, and since no local operation can increase the
coherent information of a quantum state, we obtain that for every CP-map \(\mathcal{I}_x\),
    \begin{align*}
        I_{\sigma^x_{RQ}}(R\rangle Q)=S_{\sigma^x}(Q)-S_{\sigma^x}(RQ)=\log_2
	d_{in},
    \text{ where
	}\sigma^x_{RQ}=\frac{(id\otimes\mathcal{I}_x)(\ketbra{\psi}{\psi}_{RA})}{\Tr[(id\otimes\mathcal{I}_x)(\ketbra{\psi}{\psi}_{RA})]}.
    \end{align*}
To use a standard argument, let $\varphi^x_{RQE}$ be a purification of $\sigma^x_{RQ}$, then
we have
\[
\log_2d_{in}=I_{\sigma^x_{RQ}}(R\rangle
Q)=S_{\varphi^x}(RE)-S_{\varphi^x}(E)=S_{\varphi^x}(R|E)\le S_{\varphi^x}(R)\le \log_2
d_R\le \log_2 d_{in},
\]
so that all inequalities are equalities. It follows that  \(d_A=d_{in}=d_R\) and
\(\varphi^x_R=\sigma^x_R\) is
the maximally mixed state, hence the input quantum state \(\ket{\psi}_{RA}\) must be a
maximally entangled state (say \(\ket{\phi^+}_{RA}\)) in \(\mathcal{H}_{in}^{\otimes 2}\).
Moreover, the equality \(S_{\varphi^x}(R|E)= S_{\varphi^x}(R)\) implies that
\(\varphi^x_{RE}=\varphi^x_R\otimes \varphi^x_E\), so that the purification
\(\varphi^x_{RQE}\) has the form
\[
\varphi^x_{RQE}=(I_{R}\otimes U^x\otimes I_E)(\ketbra{\phi^+}{\phi^+}_{RA}\otimes
\varphi^x_{E'E})(I_{R}\otimes (U^x)^\dagger\otimes I_E)
\]
for some purification \(\varphi^x_{E'E}\) of \(\varphi^x_E\) and a unitary \(U^x:
\mathcal{H}_{AE'}\to \mathcal{H}_Q\). Taking the partial trace over $E$, we obtain
\[
\sigma^x_{RQ}=\varphi^x_{RQ}=(I_{R}\otimes U^x)(\ketbra{\phi^+}{\phi^+}_{RA}\otimes
\varphi^x_{E'})(I_{R}\otimes (U^x)^\dagger).
\]
Let \(\varphi^x_{E'}=\sum_l \lambda^x_l\ketbra{\xi^x_l}{\xi_l^x}\) be the spectral
decomposition, and let
$p_x:=\Tr[(id\otimes\mathcal{I}_x)(\ketbra{\psi}{\psi}_{RA})]$.  Putting all together, we obtain that
\[
(id\otimes \Lambda_{\text{ind}})(\ketbra{\phi^+}{\phi^+}_{RA})=(id\otimes
\Lambda_{\text{ind}})(\ketbra{\psi}{\psi}_{RA})=\sum_xp_x\sigma^x_{RQ}=\sum_{x,l} p_x\lambda^x_l (I_{R}\otimes
\mathcal{V}_l^x)(\ketbra{\phi^+}{\phi^+}_{RA})(I_{R}\otimes (\mathcal{V}_l^x)^\dagger),
\]
where \(\mathcal{V}^x_l\) is an isometry \(\mathcal{H}_A\to\mathcal{H}_Q\) acting as
\[
\mathcal{V}^x_l\ket{\xi}_A=U^x(\ket{\xi}_A\otimes \ket{\xi^x_l}_{E'}),\qquad \ket{\xi}\in
\mathcal{H}_A.
\]
By uniqueness of the Choi representation, the induced channel 
\(\Lambda=\sum_{x,l}p_x\lambda^x_l\mathcal{V}^x_l\cdot
(\mathcal{V}^x_l)^\dagger\) is a convex combination of isometric channels.
\end{proof}\noindent
The above proposition can also be seen as the consequence of error correction for quantum channels \cite{gregoratti2004quantum,nayak2006invertible}.\\
\noindent
For perfect demultiplexing of the classical information, the information encoded in the
input quantum state must be stored in the classical register. In the following, we will
identify quantum channels that have a Q-DEMUX implementation attaining maximal classical capacity.

\begin{proposition}\label{p2} A quantum channel \(\Lambda_{\text{ind}}\) has a Q-DEMUX realization
\(\mathrm{D}_0\) with  \(\mathcal{C}(\mathrm{D}_0)=\log_2 d_{\text{in}}\) if and only if \(\Lambda\) is a classical-quantum (C-Q) channel.
\end{proposition}
\begin{proof}
To prove the \textit{if}-part, let us consider a classical-quantum channel \(\Lambda\), involving a projective measurement \(\{\ketbra{k}{k}\}_k\) on the input quantum state and preparing the corresponding output quantum states \(\{\sigma_k\in \mathcal{D}(\mathcal{H}_{out})\}_k\). That is
\[\Lambda(\rho)=\sum_k\langle k|\rho|k\rangle\sigma_k,\quad\forall\rho\in\mathcal{D}(\mathcal{H}_{in}).\]
Now one can simply consider a quantum instrument \(\mathcal{I}(\Omega=[k], \mathcal{H}_{in},\mathcal{H}_{out})\),  where the  CPTNI operations and the induced POVM are given as
\begin{equation}\label{eq:instrument_CQ}
\mathcal{I}_x=\bra{x}\cdot \ket{x}\sigma_x,\qquad 
\mathcal{A}_x^{\mathcal{I}}\equiv\ketbra{x}{x},\qquad  x\in [k].
\end{equation}
Accordingly, if Alice encodes the input classical random variable \(y\in Y\) in the orthogonal set of quantum states \(\{\ketbra{y}{y}\in\mathcal{D}(\mathcal{H}_{in})\}_y\), then
\(p(Y=y|X=x)=\Tr[\mathcal{A}_x^{\mathcal{I}}\ketbra{y}{y}]=\delta_{x,y}\). 
Now, maximizing over the input probability distribution \(\{p_y\}_y\), Alice is able to communicate \(\log_2d_{in}\)-bit of classical information for the completely random one, i.e., \(p_y=\frac1{d_{in}},~\forall y\in Y\).
    
To prove the \textit{only if}-part, let us consider a quantum instrument \(\mathcal{I}(\Omega,\mathcal{H}_{in},\mathcal{H}_{out})\) with the induced POVM \(\{\mathcal{A}^{\mathcal{I}}_x\}_{x\in\Omega}\) and induced CP-maps \(\{\mathcal{I}_x\}_{x\in\Omega}\) achieving the optimal classical capacity. In other words, there exists a a set of \(d_{in}\) quantum states \(\{\rho_y\in\mathcal{D}(\mathcal{H}_{in})\}_y\) encoded with the classical random variable \(y\in Y\), such that 
\[\Tr[\mathcal{A}^{\mathcal{I}}_x\rho_y]=\delta_{x,y}.\]
This immediately implies, \(\{\rho_y=\ketbra{y}{y}\}_y\) forms an orthogonal basis for \(\mathcal{H}_{in}\) and \(\{\mathcal{A}^{\mathcal{I}}_x\}_x\) is equivalent to the \(d_{in}\) outcome projective measurement \(\{\ketbra{x}{x}\in\mathcal{D}(\mathcal{H}_{in})\}_x\).  Since 
\(\mathcal{I}_x^*(I_{\text{out}})=\mathcal{A}^{\mathcal{I}}_x=\ketbra{x}{x}\) and \(\mathcal{I}_x^*\) is positive,
we have
\[
\mathcal{I}^*_x(B)=\omega_x(B)\ketbra{x}{x},\qquad \forall B\in
\mathcal{L}(\mathcal{H}_{\text{out}}),\quad  0\le
B\le I_{\text{out}}
\]
for some \(\omega_x(B)\in [0,1]\). By linearity of \(\mathcal{I}_x^*\), the map
\(B\mapsto\omega_x(B)\)  from the set of effects \(\{B\in
\mathcal{L}(\mathcal{H}_{\text{out}})\ |\ 0\le B\le I_{\text{out}}\}\) into the interval \([0,1]\)
preserves convex combinations, and therefore  it extends to a
positive linear functional over \(\mathcal{L}(\mathcal{H}_{\text{out}})\), moreover,
\(\omega_x(I_{\text{out}})=1\). It follows that there is some density matrix \(\sigma_x\in
\mathcal{D}(\mathcal{H}_{\text{out}})\) such that \(\omega_x(B)=\Tr[\sigma_xB]\),
\(\forall B\in \mathcal{L}(\mathcal{H}_{\text{out}})\). 
This implies that 
\[
\mathcal{I}_x(\rho)=\bra{x}\rho\ket{x}\sigma_x,\qquad \rho\in
\mathcal{D}(\mathcal{H}_{\text{in}}),\ x\in \Omega,
\]
so that the induced channel \(\Lambda=\sum_x \mathcal{I}_x\) is
classical-quantum.
This completes the proof.
\end{proof}

Using Proposition \ref{p1} and Proposition \ref{p2}, we can now conclude that for a
certain class of quantum channels, Q-DEMUX realizations exist such that the total
demultiplexing strength attains the maximal possible value.

\begin{theorem}\label{t1}
   A quantum channel \(\Lambda_{\text{ind}}\) admits a  A Q-DEMUX realization \(\mathrm{D}_s\) with
   \(\mathcal{T}_s(\mathrm{D}_s)=2\log_2d_{in}\), if and only if \(\Lambda_{\text{ind}}\) is both
   classical-quantum and a random isometry (CQ-RI).
\end{theorem}
\begin{proof}
The \textit{only if}-part directly follows from the above two propositions. For the
converse,  consider the quantum channel of the form
\[\Lambda_{\text{ind}}(\rho)=\sum_kp_k\mathcal{V}_k\rho\mathcal{V}_k^{\dagger}=\sum_j\langle j|\rho|j\rangle\sigma_j,\]
    where \(\{p_k\}_k\) is a probability distribution, for every \(k\), \(\mathcal{V}_k:\mathcal{H}_{in}\to\mathcal{H}_{out}\) is an isometry and \(\{\ket{j}\}_j\) forms an orthonormal basis for \(\mathcal{H}_{in}\), with \(\sigma_j\in\mathcal{D}(\mathcal{H}_{out})\) for every \(j\). 
Consider the instrument realizations in the proofs of  Proposition \ref{p1} resp. Proposition \ref{p2} and let us construct the
Q-DEMUX such that these instruments correspond to the selector choices \(s=1\) resp.\
\(s=0\). Then  we know that with the maximally entangled state \(\ket{\phi^+}\), we have 
 \(\mathcal{Q}(\mathrm{D}_1,\ket{\phi^+}_{RA})=\log_2d_A\), and  with the CQ-state
\(\rho_{RA}=\frac{1}{d_A}\sum_{y=0}^{d_A-1}\ketbra{y}{y}\otimes\ketbra{y}{y}\), we have
\(\mathcal{C}(\mathrm{D}_0,\rho_{RA})=\log_2d_A\). Since \(\rho_A=\phi^+_A\) is the
maximally mixed state, we obtain
\[
\mathcal{T}_s(\mathrm{D}_s)=\mathcal{C}(\mathrm{D}_0,\rho_{RA})+\mathcal{Q}(\mathrm{D}_1,\phi^+_{RA})=2\log_2d_A.
\]

\end{proof}

For the sake of completeness, we will write down an isometric extension of a Q-DEMUX
\(\mathrm{D}_s\) with
maximal total demultiplexing strength. As we have seen, the induced
channel of \(\mathrm{D}_s\) must have the form 
\[
\Lambda_{\text{ind}}(\rho)=\sum_kp_k\mathcal{V}_k\rho\mathcal{V}_k^{\dagger}=\sum_j\langle j|\rho|j\rangle\sigma_j,
\]
as in Theorem \ref{t1}. The C-Q form of \(\Lambda_{\text{ind}}\) suggests  a minimal
Stinespring dilation with an isometry of the form
\begin{equation}\label{eq:mst}
\mathcal{V}_{\Lambda}^{CQ}\ket{\psi}_A=\sum_jc_j\ket{\sigma_j}_{QC_{1,j}}\otimes\ket{j}_{C_2},\qquad \forall \ket{\psi}=\sum_ic_i\ket{i}\in\mathcal{H}_{in},
\end{equation}
where \(\ket{\sigma_j}\) is a  purification for the quantum state
\(\sigma_j\) with ancillary space \(\mathcal{H}_{C_{1,j}}\) such that
\(d_{C_{1,j}}=\mathrm{rank}(\sigma_j)\), and \(d_{C_2}=d_A\). The random isometry form of
the channel gives an isometric extension
\[
\mathcal{V}_{\Lambda}^{RI}\ket{\psi}_A=\sum_k\sqrt{p_k}\mathcal{V}_k\ket{\psi}_Q\otimes\ket{k}_C,
\]
where the dimension \(d_C\) is the number of isometries \(\mathcal{V}_k\) appearing in the decomposition of
\(\Lambda_{\text{ind}}\) with nonzero probability \(p_k\). By minimality of the 
Stinespring dilation, we have \(d_C\ge \sum_j d_{C_{1,j}}\), and there is an isometry \(U:
\oplus_j \mathcal{H}_{C_{1,j}}\to \mathcal{H}_C\) such that 
\[
(I_{Q}\otimes U)\mathcal{V}^{CQ}=\mathcal{V}^{RI}.
\]
 We may then define the isometric dilation of the Q-DEMUX as 
\((\mathcal{H}_C\otimes \mathbb{C}^2, \mathcal{V}^{CQ}\otimes I_2, \{M_k\}_k)\),
where
\[
M_k= N_k\otimes \ketbra{0}{0}+ U^\dagger\ketbra{k}{k}U\otimes\ketbra{1}{1},
\]
with \(N_k=I_{C_{1,k}}\otimes \ketbra{k}{k}\) for \(k=0,\dots, d_A-1\) and
\(N_k=0\) for \(k\ge d_A\).

Observe that, as we have seen in the above paragraph,  the number of isometries in the induced channel of a
Q-DEMUX with maximal strength is always lower bounded by the sum of the ranks of the
densities \(\sigma_j\) in its C-Q form (which is easily seen to be equal to the Choi rank
of the induced channel). In  the example below, and for pure-perfect
Q-DEMUXs treated in the next paragraph, this lower bound is attained.  

\begin{example}(Depolarizing channel) An example of a channel accepting a Q-DEMUX
realization with maximal strength is \(\Lambda(\rho)=\Tr[\rho]\tau_A\), where
\(\tau_A=\frac{1}{d_A}I_A\) is the maximally mixed state.
Indeed,
\[
\Lambda(\rho)=\sum_j\bra{j}\rho\ket{j} \tau_A=\frac{1}{d_A^2}\sum_{k=0}^{d_A^2-1}
\mathcal{W}_k\rho\mathcal{W}_k^\dagger,
\]
where \(\mathcal{W}_k\), \(k=0,\dots,d^2_A-1\) are generalized Pauli operators
\[
\mathcal{W}_{m+d_An}=X^mZ^n,\quad m,n=0,\dots,d_A-1,\qquad
X\ket{j}=\ket{j+1\,(\mathrm{mod}\, d_A)},\
Z\ket{j}=e^{i2\pi j/d_A}\ket{j}.
\]
Note that in this case, the number of isometries $d_A^2$ is the same as the Choi rank of
\(\Lambda\).

\end{example}

It is important to note that the form of the quantum channel, in itself, is not sufficient to attain the maximum demultiplexing strength for the Q-DEMUX. Rather, the Q-DEMUX with the choice of the selector \(s\), should be able to implement a pair of specific instrument realizations for the induced quantum channel. In particular, with a CQ-RI induced quantum channel, the Q-DEMUX must admit its CQ instrument realization as the minimal dilation one, which can be used with \(s=0\) to process the classical information through it. On the other hand, in the question of quantum information transmission, by setting the selector \(s=1\), the sender can apply a specific unitary on the ancillary quantum system to change it to the RI instrument realization of the induced quantum channel.  At this end, we provide an simple circuit realization for a class of perfect Q-DEMUXs.

\subsection{Circuit Realizations for Pure-Perfect Q-DEMUXs}
While Theorem \ref{t1} provides a schematic and theoretical modeling for the perfect Q-DEMUXs, it does not specify the exact dimension of ancillary quantum systems to simulate it in a quantum circuit model. We will now consider a specific class of such perfect Q-DEMUXs, for which the ancillary quantum system can be characterized in terms of the dimension of the output and the input Hilbert spaces.  
\begin{definition}\label{d1}
    A Q-DEMUX \(\mathrm{D}_s:\mathcal{D}(\mathcal{H}_{in})_A\to\mathcal{D}(\mathcal{H}_{out})_Q\times X_C\) is said to be perfect and pure if \(\mathcal{T}_s(\mathrm{D}_s)=2\log_2d_{in}\) and there exists an orthonormal basis \(\{\ket{i}\}_i\in\mathcal{H}_{in}\) for which each of the states \(\{\Lambda_{\text{ind}}(\ketbra{i}{i})\}_i\) is pure.
\end{definition}
Before going to the circuit realization, we will highlight a specific feature of the pure-perfect Q-DEMUXs, which will be instrumental for the rest.
\begin{lemma}\label{l4}
    The induced quantum channel \(\Lambda_{\text{ind}}\) for any pure-perfect Q-DEMUX admits a realization with exactly \(d_{in}\) numbers of isometries, each with probability \(\frac1{d_{in}}\).
\end{lemma}
\begin{proof}
For a pure and perfect Q-DEMUX \(\mathrm{D}_s\), the induced quantum channel \(\Lambda_{\text{ind}}\) will take the following form:
\[\Lambda_{\text{ind}}(\rho)=\sum_kp_k\mathcal{V}_k\rho\mathcal{V}_k^{\dagger}=\sum_{j=0}^{d_{in}-1}\bra{\psi_j}\rho\ket{\psi_j}\ketbra{\phi_j}{\phi_j}.\]
Equivalence between these two realization implies that for every \(k\), \[\mathcal{V}_k\ket{\psi_j}=e^{i\theta_{kj}}\ket{\phi_j},\forall j\in\{0,\cdots,d_{in}-1\}.\]
Moreover, since every isometry preserves orthogonality and the set of states \(\{\ket{\psi_j}\}_{j=0}^{d_{in}-1}\) is mutually orthogonal, the set of states \(\{\ket{\phi_j}\}_{j=0}^{d_{in}-1}\) must also be mutually orthogonal.

Let us now consider the dilation of the induced quantum channel \(\Lambda_{\text{ind}}\) in the C-Q form:
\begin{align}\label{e8}
\mathcal{V}_{\Lambda}\ket{\psi}_A=\sum_{j=0}^{d_{in}-1}c_j\ket{\phi_j}_Q\otimes\ket{\psi_j}_C
\end{align}
for every quantum state \(\sum_jc_j\ket{\psi_j}=\ket{\psi}_A\in\mathcal{H}_{in}\).
We can further consider the Fourier transformed basis \(\{\ket{\tilde{\psi_j}}\}_{j=0}^{d_{in}-1}\) for the \(d_{in}\)-dimensional space, such that
\begin{align}\label{e8.5}
\ket{\tilde{\psi_j}}=\frac1{\sqrt{d_{in}}}\sum_ke^{\frac{2\pi (j.k)}{d_{in}}}\ket{\psi_k}\text{ and }\ket{\psi_j}=\frac1{\sqrt{d_{in}}}\sum_ke^{-\frac{2\pi (j.k)}{d_{in}}}\ket{\tilde{\psi_k}},~\forall j.\end{align}
Then by replacing the basis for the subsystem \(C\) in Eq. (\ref{e8}), we have
\begin{align}\label{e9}
\nonumber\mathcal{V}_{\Lambda}\ket{\psi}_A&=\frac1{\sqrt{d_{in}}}\sum_{j,k=0}^{d_{in}-1}c_je^{-\frac{2\pi (j.k)}{d_{in}}}\ket{\phi_j}_Q\otimes\ket{\tilde{\psi_k}}_C\\&=\frac1{\sqrt{d_{in}}}\sum_{k=0}^{d_{in}-1}(\sum_{j=0}^{d_{in}-1}c_je^{-\frac{2\pi (j.k)}{d_{in}}}\ket{\phi_j})_Q\otimes\ket{\tilde{\psi_k}}_C
\end{align}
We define a collection of isometries
\(\{\mathcal{W}_k:\mathcal{H}_{in}\to\mathcal{H}_{out}\}_{k=0}^{d_{in}-1}\) as
\begin{align}\label{e10}
    \mathcal{W}_k\ket{\psi_j}=e^{-\frac{2\pi(j.k)}{d_{in}}}\ket{\phi_j},~\forall j\in\{0,\cdots, d_{in}-1\}.
\end{align}
Using these isometries, we can rewrite Eq. (\ref{e9}) as
\[\mathcal{V}_{\Lambda}\ket{\psi}_A=\frac1{\sqrt{d_{in}}}\sum_{k=0}^{d_{in}-1}\mathcal{W}_k\ket{\psi}_Q\otimes\ket{\tilde{\psi_k}}_C.\]
By tracing out the ancillary subsystem \(C\), we obtain
\[\Lambda_{ind}(\rho)=\Tr_C[\mathcal{V}_{\Lambda}\rho\mathcal{V}_{\Lambda}^{\dagger}]=\frac1{d_{in}}\sum_{k=0}^{d_{in}-1}\mathcal{W}_k\rho\mathcal{W}_k^{\dagger}\]
and this completes the proof. 
\end{proof}
We will now move to the circuit realization for the pure-perfect Q-DEMUXs.
\begin{figure}[h]
    \centering
    \includegraphics[width=0.65\linewidth]{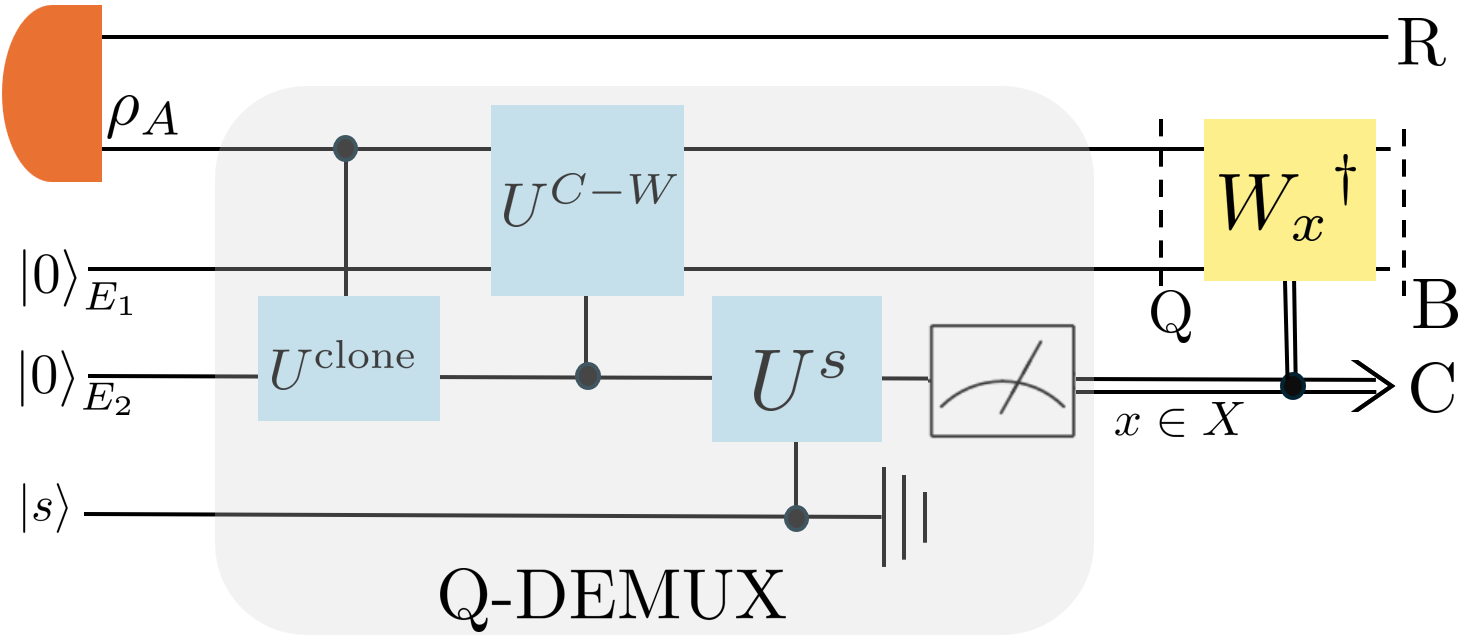}
    \caption{\emph{Circuit diagram for a pure-perfect Q-DEMUX.} The \textcolor{gray}{gray} box denotes the Q-DEMUX realization in terms of the controlled unitary gates (\textcolor{cyan}{cyan}) and the \(d_{in}\)-dimensional computational basis measurement on \(E_2\). Finally, the correcting unitary (\textcolor{YellowOrange}{yellow}) \(\{W_x^{\dagger}\}_x\) is appended on the subsystem \(Q\simeq AE_{1}\) externally.}
    \label{fig:qdemux}
\end{figure}
 \begin{proposition}\label{p3}
   Any pure-perfect Q-DEMUX \(\mathrm{D}_s:\mathcal{D}(\mathcal{H}_{in})\to\mathcal{D}(\mathcal{H}_{out})\times X\) can be simulated with a \(\lceil\frac{d_{out}}{d_{in}}\rceil\times d_{in}\)-dimensional ancillary quantum system, a selector qubit \(\ket{s}\in\{\ket{0},\ket{1}\}\) and a selector independent computational basis measurement.
\end{proposition}
\begin{proof}
   The proof is evident from Lemma \ref{l4}.
    
    Let us consider the unitary realization for the isometry \(\mathcal{V}_{\Lambda}\) in Eq. (\ref{e8}), with an ancillary quantum system initialized in the quantum state \(\ket{0}_E=\ket{0}_{E_1}\otimes\ket{0}_{E_2}\) with \(dim(\mathcal{H}_{E_1})=\lceil\frac{d_{out}}{d_{in}}\rceil\) and \(dim(\mathcal{H}_{E_2})=d_{in}\), such that
    \begin{gather}\nonumber
        U_{\Lambda}\ket{\psi}_A\otimes\ket{0}_{E}=U_{AE_1E_2}^{C-W}\circ(U_{AE_2}^{\text{clone}}\otimes I_{E_1})\ket{\psi}_A\otimes\ket{0}_{E_1}\otimes\ket{0}_{E_2},\text{ where, }\\\nonumber
      U_{AE_2}^{\text{clone}}\ket{\psi_j}_A\otimes\ket{0}_{E_2}=
       \ket{\psi_j}_A\otimes\ket{\psi_j}_{E_2}~\forall j\in\{0,\cdots,d_{in}-1\}~\text{and}\\\nonumber
       U_{AE_1E_2}^{C-W}\ket{\psi}_A\otimes\ket{0}_{E_1}\otimes\ket{\psi_j}_{E_2}= W_j(\ket{\psi}\otimes\ket{0})\otimes\ket{\psi_j}_{E_2},~\forall \ket{\psi},\ket{\psi_j}\in\mathcal{H}_{in}\text{ and }j\in\{0,\cdots,d_{in}-1\}.
    \end{gather}
    Here \(U^{\text{clone}}\in\mathcal{L}(\mathcal{H}_{in}\otimes\mathcal{H}_{E_2})\) denotes the cloning unitary for the orthonormal basis \(\{\ket{\psi_j}\}_{j=0}^{d_{in}-1}\) and \(U^{C-W}=\sum_{j=0}^{d_{in}-1}W_j\otimes\ketbra{j}{j}_{E_2}\) denotes the controlled \(W\) gate, where each unitary \(W_j:\mathcal{H}_{in}\otimes\mathcal{H}_{E_1}\to\mathcal{H}_{out}\) corresponds to the dilation of the isometries in Eq. (\ref{e10}).

Followed by this global unitary \(U_{\Lambda}\), we consider a selector dependent unitary \(U_s\in\mathcal{L}(\mathcal{H}_{E_2}\otimes\mathbb{C}^2)\), such that
\begin{align}
    U_s=U^0_{E_2}\otimes\ketbra{0}{0}_{s}+U^1_{E_2}\otimes\ketbra{1}{1}_{s},\text{ with }U^0\ket{\psi_j}=\ket{j}~\&~U^1\ket{\tilde{\psi_j}}=\ket{j},~\forall j\in\{0,\cdots, d_{in}-1\},
\end{align}
where \(\{\ket{\tilde{\psi_j}}\}\) is the Fourier Transformed basis of the orthonormal basis \(\{\ket{\psi_j}\}\), as in Eq. (\ref{e8.5}). Finally, the outcome of a computational basis measurement on the ancillary quantum system \(E_2\) is stored as the output classical random variable in \(C\) (see Fig. \ref{fig:qdemux} for the circuit).
\end{proof}
   Theorem \ref{t1} characterizes the set of quantum channels, namely the CQ-RI, for which there exists a of Q-DEMUX realization attaining the maximal demultiplexing strength. In other words, a Q-DEMUX perfectly accomplishes the task of generalized information demultiplexing, if and only if the induced quantum channel is CQ-RI and the selector can be used to implement the RI and C-Q instrument realization of the induced quantum channel to process the quantum and the classical data respectively. While the perfect Q-DEMUX configurations are completely characterized above, it stimulates further investigation for the imperfect cases. To this goal, later in Theorem \ref{t2}, we will highlight the characteristic features of the pair of instrument realizations for a given quantum channel, configuring a Q-DEMUX with demultiplexing strength beyond a critical bound.
%
%
%
%
\section{Q-DEMUXs Without Access to the Selector}
We will now consider a stronger variant of Q-DEMUX, where the sender, Alice, remains oblivious regarding the nature of the data she is supposed to process. In other words, the exact correlation between the her subsystem with that of the reference, that is, whether it is entangled or classically correlated, is unknown to Alice. Evidently, in such a scenario the access to the selector provides no additional advantage to Alice. Therefore, given a Q-DEMUX, without any access to the selector, i.e., a particular instrument realization of the induced quantum channel \(\Lambda_{\text{ind}}:\mathcal{D}(\mathcal{H}_{in})\to\mathcal{D}(\mathcal{H}_{out})\), Alice is only allowed to maximize over all possible input quantum states to achieve the maximal demultiplexing strength of that particular Q-DEMUX. Mathematically, it reads
\begin{align}\label{e13}
\mathcal{T}(\mathrm{D})=\max_{\rho_A\in\mathcal{D}(\mathcal{H}_{in})}[\mathcal{C}(\mathrm{D},\rho_{RA})+\mathcal{Q}(\mathrm{D},\ket{\psi}_{RA})]
\end{align}
where, \(\rho_{RA}\) and \(\ket{\psi}_{RA}\) are the same as mentioned in Eq. (\ref{eq:e4}). However, in contrast to Eq. (\ref{eq:e4}), here the selector subscript \(s\in\{0,1\}\) is omitted from the expression, denoting the scenario of no access to the selector. Additionally, with a further maximization over all Q-DEMUXs with the same induced quantum channel \(\Lambda_{\text{ind}}\), we can characterize the demultiplexing strength of the channel itself, namely
\begin{align}\label{e13.5}\mathcal{T}(\Lambda_{\text{ind}})=\max_{\substack{\mathrm{D}:\mathcal{D}(\mathcal{H}_{in})\to\mathcal{D}(\mathcal{H}_{out})\times X_C\\s.t.,\Tr_C[\mathrm{D}(\rho_A)]= \Lambda_{\text{ind}}(\rho_A)}}\mathcal{T}(\mathrm{D}).\end{align}

\noindent
Nevertheless, we will now derive an upper-bound for demultiplexing strength of a Q-DEMUX, where the sender has no access to the selector, for which the following lemma is instrumental.
\begin{lemma}\label{l4.5}
    If a quantum channel \(\Lambda_{\text{ind}}\) achieves its maximal demultiplexing strength \(\mathcal{T}(\Lambda_{\text{ind}})\) (as in Eq. (\ref{e13.5})) for a Q-DEMUX configuration \(\mathrm{D}^*\), then the instrument realization of \(\Lambda_{\text{ind}}\) in \(\mathrm{D}^*\), can be sufficiently modeled by the rank one measurement on the ancillary quantum system.
\end{lemma}
\begin{proof}
    By contradiction, suppose the instrument realization \(\mathcal{I}(X, \mathcal{H}_{in},\mathcal{H}_{out})\) for \(\mathrm{D}^*\) involves a POVM \(\{M_x\}_{x=0}^{X-1}\) on the ancillary quantum system, such that for one of the POVM elements (say \(M_{X-1},\)), \(rank(M_{X-1})>1\). This further implies 
    \[M_{X-1}\equiv\sum_{z=0}^{Z-1} N_z,\text{ where }~\forall z,~rank(N_z)=1,~\&~0<N_z<I.\]
    Moreover, we assume that the classical-quantum state \(\rho_{RA}\) and the entangled state \(\ket{\psi}_{RA}\), such that \(\Tr_R(\rho_{RA})=\Tr_R(\ketbra{\psi}{\psi}_{RA})=\rho_A\), along with the correcting operations \(\{\mathcal{N}_x\}_x\) maximizes \(\mathcal{T}(\mathrm{D}^*)\). 
    
    Then, consider another Q-DEMUX \(\mathrm{D}'\) in terms of the instrument realization \(\mathcal{I}'(X',\mathcal{H}_{in},\mathcal{H}_{out})\) of \(\Lambda_{\text{ind}}\), implementing the one-rank POVM \(\{\{M_x\}_{x\neq X-1},\{N_z\}_z\}\) on the ancillary quantum system. Moreover, if we denote \(X':=\{0,\cdots,~X-2,~X-1,\cdots, X+Z-2\}\), then evidently \[\mathcal{I}'_x\equiv\mathcal{I}_x,~\forall x\in\{0,\cdots,X-2\}\text{ and }\sum_{x=X-1}^{X+Z-2}\mathcal{I}'_x=\mathcal{I}_{X-1}.\]
    Therefore, for an input quantum state \(\rho\in\mathcal{D}(\mathcal{H}_{in})\),
    \begin{align}
     \nonumber\mathrm{D}^*(\rho)&=\sum_{x=0}^{X-1}\mathcal{I}_x(\rho)\otimes\ketbra{x}{x},\\\nonumber
 \text{ while, }\mathrm{D}'(\rho)&=\sum_{x=0}^{X-2}
\mathcal{I}_x(\rho)\otimes\ketbra{x}{x}+ \sum_{z=X-1}^{X+Z-2}
\mathcal{I}'_z(\rho)\otimes\ketbra{z}{z}.
    \end{align}
    For the classical communication purpose, suppose Alice input the same state \(\rho_A=\sum_yp_y\rho_y\) to \(\mathrm{D}'\), sharing the same classical correlation \(\rho_{RA}\) with \(R\). The output random variable accordingly decoded as \(x\) for \(x\in\{0,\cdots,X-2\}\) and \(X-1\) whenever \(x\in\{X-1\cdots,X+Z-2\}\). This implies that
    \begin{align}\label{e5.1}
    \mathcal{C}(\mathrm{D}',\rho_{RA})=\mathcal{C}(\mathrm{D}^*,\rho_{RA}).
    \end{align}
    On the other hand, to process the quantum information consider the same \(\ket{\psi}_{RA}\), with the marginal \(\rho_A\) as the input to \(\mathrm{D}'\) and the correcting operations are \(\{\mathcal{N}_x\}_x\) for \(x\in\{0,\cdots, X-2\}\) and \(\mathcal{N}_{X-1}\) for \(x\in\{X-1,\cdots,X+Z-2\}\). Then
    \begin{align}\label{e5.2}
    \mathcal{Q}(\mathrm{D}',\ket{\psi}_{RA})=\max_{\{\mathcal{M}_x\}_x}I_{\sigma(\psi)}(R\rangle B)\geq \mathcal{Q}(\mathrm{D}',\ket{\psi}_{RA},\{\mathcal{N}_x\}_x)=\mathcal{Q}(\mathrm{D}^*,\ket{\psi}_{RA}),
    \end{align}
    where, \(\sigma(\psi)\) is the final quantum state shared between \(R\) and \(B\), after passing \(\Tr_R(\ketbra{\psi}{\psi}_{RA})\) through \(\mathrm{D}'\) and corrected by the optimized correcting operations. Also, \(\mathcal{Q}(\mathrm{D}',\ket{\psi}_{RA},\{\mathcal{N}_x\}_x)\) denotes the quantum capacity for a specific choice of the correcting operations \(\{\mathcal{N}_x\}_x\).

    Using Eqs. (\ref{e5.1}) and (\ref{e5.2}) together, we can conclude
    \begin{align*}
        \mathcal{T}(\mathrm{D}')&=\max_{\rho'_A\in\mathcal{D}(\mathcal{H}_{in})}[\mathcal{C}(\mathrm{D}',\rho'_{RA})+\mathcal{Q}(\mathrm{D}',\ket{\psi'}_{RA})]\\&\geq\mathcal{C}(\mathrm{D}',\rho_{RA})+\mathcal{Q}(\mathrm{D}',\ket{\psi}_{RA},\{\mathcal{N}_x\}_x)\\
        &=\mathcal{C}(\mathrm{D}^*,\rho_{RA})+\mathcal{Q}(\mathrm{D}^*,\ket{\psi}_{RA})=\mathcal{T}(\mathrm{D}^*).
    \end{align*}
    But, since both \(\mathrm{D}^*\) and \(\mathrm{D}'\) possesses the same induced channel \(\Lambda_{\text{ind}}\), this contradicts our initial assumption that \(\mathcal{T}(\Lambda_{\text{ind}})=\mathcal{T}(\mathrm{D}^*)\). Hence completes the proof.
\end{proof}
At this point, it is important to mention that any POVM measurement can be classically post-processed from another POVM measurement, for which each of the elements are rank one. Therefore, Lemma \ref{l4.5} operationally implies that the classical post-processing of the quantum instrument can not increase its demultiplexing strength.
\begin{theorem}\label{t2}
    For any Q-DEMUX \(\mathrm{D}:\mathcal{D}(\mathcal{H}_{in})\to\mathcal{D}(\mathcal{H}_{out})\times X\), the selector-less total demultiplexing strength \(\mathcal{T}(\mathrm{D})\leq \log_2 d_{in}\).
\end{theorem}
\begin{proof}

For a given induced quantum channel \(\Lambda_{\text{ind}}:\mathcal{D}(\mathcal{H}_{in})\to\mathcal{D}(\mathcal{H}_{out})\),
let us consider a Q-DEMUX \(\mathrm{D}\) in terms of an instrument realization
\(\mathcal{I}(X,\mathcal{H}_{in},\mathcal{H}_{out})\), involving the measurement
\(\{M_x\}_x\) on the ancillary quantum system (\(E\)). 

To establish the upper bound, using Lemma \ref{l4.5} we may therefore assume that each \(M_x\) is rank-one.

Next, for an input quantum state \(\rho_A\in\mathcal{D}(\mathcal{H}_{in})\), let 
\[
\rho_{AB}=\sum_yp_y\ketbra{y}{y}_R\otimes\rho^y_A,~\&~\rho_A=\Tr_R[\rho_{RA}],\quad\text{
and }
\sigma_{RQX}=\sum_yp_y\ketbra{y}{y}_R\otimes\mathcal{I}_x(\rho^y_A)\otimes\ketbra{x}{x}.
\]
The corresponding 
classical capacity is given as
\begin{equation}\label{e14}
\mathcal{C}(\mathrm{D},\rho_{RA})= I_{\sigma}(R:X)=H_\sigma(R)-H_\sigma(R|X)\end{equation}
Note that, in Eq. (\ref{e14}) we have used the Shannon entropy,
\(H(Z)=-\sum_zp_z\log_2p_z\), since the state \(\Tr_A[\rho_{RA}]\) on the reference
subsystem is classical and the \(X\) denotes the random variable in the classical output
port \(C\) of the Q-DEMUX.
To bound the quantum capacity, let \(\ket{\psi}_{RA}\) be a purification of \(\rho_A\) and let
\[
\eta_{RQX}=\sum_x (id\otimes \mathcal{I}_x)(\ketbra{\psi}{\psi}_{RA})\otimes \ketbra{x}{x}.
\]
For a given set of correction operations\(\{\mathcal{N}_x\}_x\), we denote the corrected
output state as
\[
\eta^{\mathcal{N}}_{RBX}=\sum_x (id\otimes
\mathcal{N}_x\circ\mathcal{I}_x)(\ketbra{\psi}{\psi}_{RA})\otimes \ketbra{x}{x}.
\]
We then have 
\begin{equation}\label{e15}
I_{\eta^{\mathcal{N}}}(R\rangle B)\le I_{\eta^{\mathcal{N}}}(R\rangle B|X)\le
I_{\eta}(R\rangle Q|X)=S_\eta(Q|X)-S_\eta(RQ|X)
\end{equation}
Here \(I_\eta(R\rangle Q|X)\) denotes the coherent
information of the state between \(R\) and the output quantum port \(Q\) of the Q-DEMUX,
conditioned over the classical outcome \(x\in X\) registered in \(C\).
The second inequality follows from the fact that for each conditional state
\(\eta_{RQ|x}\), applying the quantum channel \(\mathcal{N}_x\) cannot increase the
coherent information conditioned on \(x\). 
Note also that in Eq. (\ref{e15}),  the subsystem
\(R\) must be treated as a quantum system.

Now note that since \(M_x\) are rank-one, the post-measurement state \(\eta_{RQ|x}\)
between \(R\) and \(Q\) is pure for each of the outcome \(x\in X\).
Therefore, with \(p_x=\Tr[(id_R\otimes \mathcal{I}_x)(\ketbra{\psi}{\psi})]\) denoting the probability of observing \(x\) at
the classical output port,  we have \(S_\eta(RQ|X)=\sum_{x\in X}p_xS_\eta(RQ|X=x)=0\) and 
\(S_\eta(Q|X)=\sum_{x\in X}p_xS_\eta(Q|x)=\sum_{x\in X}p_xS_\eta(R|x)=S_\eta(R|X)\). It
follows that 
\[
\mathcal{Q}(\mathrm{D},\ket{\psi}_{RA})=\max_{\{\mathcal{N}_x\}_x} I_{\eta^{\mathcal{N}}}(R\rangle
B)\leq S_\eta(R|X).
\]
Finally, we obtain that 
\begin{align*}
\mathcal{C}(\mathrm{D},\rho_{RA})+\mathcal{Q}(\mathrm{D},\ket{\psi_{RA}})&\le
H_\sigma(R)-H_\sigma(R|X)+S_\eta(R|X)\\
&\le H_\sigma(R)\le \log_2 d_A,
\end{align*}
the second inequality follows from the fact that \(S_\eta(R|X)\leq H_\sigma(R|X)\). Taking
the supremum over all C-Q states \(\rho_{RA}\), purifications \(\ket{\psi_{RA}}\), and
finally over all input states \(\rho_A\)  completes the proof.
\end{proof}
It is now evident from Theorem \ref{t2}, that both for classical-quantum and the random
isometry channels there exist  instrument realization saturating the optimal bound of
demultiplexing strength even if the sender has no access to the selector. In particular,
for C-Q channels the instrument realization as in Eq. \eqref{eq:instrument_CQ} maximizes
\(\mathcal{T}(\mathrm{D})\) by \(\mathcal{C}(\mathrm{D})=\log_2d_{in}\) and
\(\mathcal{Q}(\mathrm{D})=0\). For a random isometry channel, the instrument
realization as in Eq. \eqref{eq:instrument_RI} attains
\(\mathcal{Q}(\mathrm{D})=\log_2d_{in}\) and \(\mathcal{C}(\mathrm{D})=0\), again
saturating the upper-bound for the selector-less total demultiplexing strength
\(\mathcal{T}(\mathrm{D})\). However, in contrast to Theorem \ref{t1}, we will show that
these are not the only classes of Q-DEMUXs maximizing the selector-less demultiplexing
strength. In fact, in the following we will show that any probabilistic mixture of a C-Q and
a random isometry channel has an instrument realization that achieves the optimal value for \(\mathcal{T}(\mathrm{D})\).
\begin{proposition}\label{p4}
   A quantum channel \(\Lambda_{\text{ind}}\) admits a Q-DEMUX realization \(\mathrm{D}\), for which the selector-less demultiplexing strength becomes optimal, if the channel can be represented as the direct sum of the random isometry and classical-quantum channels in two disjoint subspace of \(\mathcal{H}_{in}\).
\end{proposition}
\begin{proof}
    Let us consider a quantum channel \(\Lambda_{\text{ind}}:\mathcal{D}(\mathcal{H}_{in})\to\mathcal{D}(\mathcal{H}_{out})\) of the following form:
    \begin{align}\label{ep40}
        \Lambda_{\text{ind}}(\rho)=\sum_{j=0}^{k-1}\bra{\psi_j}\rho\ket{\psi_j}\sigma_j+\sum_{i=0}^{n-1}p_i\mathcal{V}_i[(I_{d_{in}}-\sum_{j=0}^{k-1}\ketbra{\psi_j}{\psi_j})\rho(I_{d_{in}}-\sum_{j=0}^{k-1}\ketbra{\psi_j}{\psi_j})]\mathcal{V}_i^{\dagger},
    \end{align}
    where, \(\langle\psi_m|\psi_n\rangle=\delta_{mn}\), \(\{\sigma_j\}_{j=0}^{k-1}\) with \(k<d_{in}\) is an arbitrary set of quantum states in \(\mathcal{D}(\mathcal{H}_{out})\) and for all \(l\in\{0,\cdots,n-1\}\), \(\mathcal{V}_l:\mathcal{H}_{in}\to\mathcal{H}_{out}\) is an arbitrary isometry. Note that the channel can be think of a classical-quantum channel in the \(k\)-dimensional subspace spanned by \(\{\ket{\psi_0},\cdots,\ket{\psi_{k-1}}\}\) and on the \((d_{in}-k)\)-dimensional subspace spanned by \(\{\ket{\psi_k},\cdots,\ket{\psi_{d_{in}-1}}\}\) the channel is a random isometry one.

    Now consider, the following isometry realization of the channel \(\Lambda_{\text{ind}}\), for any arbitrary \(\sum_{j=0}^{d_{in}-1}c_j\ket{\psi_j}=:\ket{\psi}\in\mathcal{H}_{in}\) as the input state
    \[\mathcal{V}_{\Lambda}\ket{\psi}_A=\sum_{j=0}^{k-1}c_j\ket{\phi_j}_{QQ'}\otimes\ket{j}_E+\sum_{j=k}^{n+k-1}\sum_{l=k}^{d_{in}-1}\sqrt{p_{j-k}}c_l(\mathcal{V}_{j-k}\ket{\psi_l})_Q\otimes\ket{\xi}_{Q'}\otimes\ket{j}_E.\]
    Here, \(\Tr_{Q'}[\ketbra{\phi_j}{\phi_j}_{QQ'}]=(\sigma_j)_Q\) for all \(j\in\{0,\cdots,k-1\}\), \(\ket{\xi}\) is an arbitrary quantum state in \(\mathcal{H}_{Q'}\) and \(\{\ket{j}\}_{j=0}^{n+k-1}\) forms an orthonormal basis for \(\mathcal{H}_E\). 
    
    Note that to construct a Q-DEMUX \(\mathrm{D}\), involving the above isometry, we will require a ancillary quantum system of dimension \(\lceil\frac{d_{out}}{d_{in}}\rceil\times \min\{d_Q,d_{Q'}\}\times (k+n)\). After the action of \(\mathcal{V}_{\Lambda}\), the sub-system \(Q'\) is traced out, sub-system \(Q\) stored at the output quantum port and the sub-system \(E\) is measured in the orthonormal basis \(\{\ketbra{j}{j}\}_{j=0}^{k+n-1}\), the outcome of which is stored at the classical register. Additionally, we define the correction operations \(\{\mathcal{N}_x\) over \(\mathcal{L}(\mathcal{H}_{out})\) depending upon the classical random variable \(x\in X\) as 
    \begin{align*}
        \mathcal{N}_x(\sigma)=\begin{cases}
             \Tr[\sigma]\ketbra{x}{x},\quad &x\in\{0,\cdots,k-1\}\\
\mathcal{V}_x^{\dagger}\sigma\mathcal{V}_x+\Tr[\sigma(I_{out}-P_x)]\omega,\quad &x\in\{k,\cdots, n+k-1\}
        \end{cases}
    \end{align*}
    where \(P_x\) and \(\omega\) are same as in the proof of the Proposition \ref{p1}.

    Now, if the \(A\) sub-system of the state \(\ket{\phi^+_{d_{in}}}\in\mathcal{H}_{in}^{\otimes 2}\) is sent through the Q-DEMUX, then after the correcting operations the final state can be written as
    \[\eta^{\mathcal{N}}_{RB}=\frac1{d_{in}}\left(\sum_{j=0}^{k-1}\ketbra{j}{j}\otimes\ketbra{j}{j}+\sum_{j,l=k}^{d_{in}-1}\ketbra{j}{l}\otimes\ketbra{j}{l}\right)_{RB}=\frac1{d_{in}}\left(\sum_{j=0}^{k-1}\ketbra{j}{j}\otimes\ketbra{j}{j}\right)_{RB}+\frac{d_{in}-k}{d_{in}}\ketbra{\phi^+_{d_{in}-k}}{\phi^+_{d_{in}-k}}_{RB}.\]
    Accordingly, \begin{align}\label{ep4.1}
    \mathcal{Q}(\mathrm{D},\ket{\phi^+_{d_{in}}})\geq I_{\eta^{\mathcal{N}}}(R\rangle B)=(1-\frac k{d_{in}})\log_2(d_{in}-k).\end{align}

    On the other hand, for classical information processing consider the classical-quantum state \(\rho_{RA}=\frac1{d_{in}}\sum_{y=0}^{d_{in}-1}\ketbra{y}{y}_R\otimes\ketbra{\psi_y}{\psi_y}_A\). Then the output classical random variable as \(x\in\{0,\cdots, n+k-1\}=: X\), the joint distribution will be the following:
    \begin{align*}
       p(x,y)&=\frac1{d_{in}},\quad y=x\in\{0,\cdots,k-1\}\\
       p(x,y)&=0,\quad y,x\in\{0,\cdots,k-1\}~\&~y\neq x\\
       p(x,y)&=\frac {p_{x-k}}{d_{in}},\quad x\in\{k,\cdots,n+k-1\}~\&~y\in\{k,\cdots,d_{in}-1\}\\
       p(x,y)&=0,\quad x\not\in\{k,\cdots,n+k-1\}~\&~y\in\{k,\cdots,d_{in}-1\}.
    \end{align*}
    With simple numerical exercise, it is now easy to show that
    \begin{align}\label{ep4.2}  
    \mathcal{C}(\mathrm{D},\rho_{RA})\geq I(Y:X)=H(Y)+H(X)-H(XY)=\log d_{in}-\frac{d_{in}-k}d\log_2(d_{in}-k).\end{align}
    Also note that both for \(\ket{\phi^+}_{RA}\) and \(\rho_{RA}\), the marginal state is identical: \(\rho_A=\frac{I_{d_{in}}}{d_{in}}\). Therefore, using Eq. (\ref{ep4.1}) and (\ref{ep4.2}), we can then conclude \[\mathcal{T}(\mathrm{D})\geq\mathcal{Q}(\mathrm{D},\ket{\phi^+_{d_{in}}}_{RA})+\mathcal{C}(\mathrm{D},\rho_{RA})=\log_2 d_{in}.\]
    The above inequality further saturates by using the upper-bound on \(\mathcal{T}(\mathrm{D})\) as derived in Theorem \ref{t2}. This completes the proof. 
\end{proof}
Note that, that the class of quantum channels we have considered in Proposition \ref{p4} implicitly assumes \(d_{in}-k\geq 2\) and hence \(d_{in}>3\) to be a nontrivial one. In particular, for qubit and qutrit systems, quantum channel of the form Eq. (\ref{ep40}) can reach to optimal selector-less demultiplexing strength whenever \(k\in\{0,d_{in}\}\), i.e., the trivial cases of RI and C-Q channels respectively. But does it imply that for \(d_{in}\leq 3\) there is no example of optimal selector-less Q-DEMUX, for which the induced quantum channel is neither RI nor C-Q? In the following, we will show that the class of induced quantum channels for \(d_{in}\geq 2\), expressed as the convex combination of a RI and C-Q channel, admits an instrument realization achieving optimal selector-less demultiplexing strength.
\begin{proposition}\label{p5}
    A quantum channel \(\Lambda_{\text{ind}}\) admits a Q-DEMUX realization \(\mathrm{D}\), for which the selector-less demultiplexing strength becomes optimal, if the channel can be represented as the convex combination of the random isometry and classical-quantum channels.
\end{proposition}
\begin{proof}
    Let us consider a quantum channel \(\Lambda_{\text{ind}}:\mathcal{D}(\mathcal{H}_{in})\to\mathcal{D}(\mathcal{H}_{out})\) of the following form:
    \begin{align*}
        \Lambda_{\text{ind}}(\rho)= p\sum_{k=0}^{K-1}p_k\mathcal{V}_k\rho\mathcal{V}_k^{\dagger}+(1-p)\sum_j\bra{\psi_j}\rho\ket{\psi_j}\sigma_j,
    \end{align*}
    where, for all \(k,~\mathcal{V}_k:\mathcal{H}_{in}\to\mathcal{H}_{out}\) is an isometry, \(\{\ket{\psi_j}\}_{j=0}^{d_{in}-1}\) forms an orthogonal basis for \(\mathcal{H}_{in}\) and for all \(j\in\{0,\cdots,d_{in}-1\},~\sigma_j\in\mathcal{D}(\mathcal{H}_{out})\). 
   
This channel has an instrument realization \(\mathcal{I}=\{\mathcal{I}_x\}_{x\in X}\) with
\(|X|=K+d_A\) of \(\Lambda_{\text{ind}}\) of the form 
\[
\mathcal{I}_x(\rho)=\begin{dcases} p \,p_x\mathcal{V}_x\rho\mathcal{V}^\dagger_x, &
x=0,\dots,K-1\\[1em]
(1-p)\bra{\psi_j}\rho\ket{\psi_j}\sigma_j, & x=K+j,\ j=0,\dots,d_A-1.
\end{dcases}
\]
We introduce the Q-DEMUXs outcome space
\(\mathcal{H}_B=\mathcal{H}_A\oplus\mathcal{H}_{A'}\) with \(A\simeq A'\), and a set of correction  operations
\(\mathcal{N}_x: \mathcal{D}(\mathcal{H}_Q)\to\mathcal{D}(\mathcal{H}_B)\) as
\[
\mathcal{N}_x(\sigma_Q)=\begin{dcases}
\mathcal{V}^\dagger_x\sigma_Q\mathcal{V}_x+\Tr[\sigma_Q(I-P_x)]\omega, & x=0,\dots,
K-1\\[1em]
\Tr[\sigma_Q]\ketbra{\phi_j}{\phi_j}_{A'}, & x=K+j,\ j=0,\dots, d_A-1,
\end{dcases}
\]
where \(P_x\) and \(\omega\) are as in the proof of Proposition \ref{p1} and
\(\{\ket{\phi_j}_{A'}\}_j\) is any orthonormal basis of the space \(\mathcal{H}_{A'}\). 
Now, consider the completely mixed input state \(\rho_A=\frac1{d_{in}}I\). 
It is easily checked that using its  C-Q
extension \(\rho_{RA}=\frac1{d_{in}}\sum_j \ketbra{j}{j}_R\otimes
\ketbra{\psi_j}{\psi_j}_A\) for classical information processing, we get
\(\mathcal{C}(\mathrm{D},\rho_{RA})=(1-p)\log_2d_A\). As for the quantum information
processing, after applying the instrument and corrections to the maximally entangled
state \(\ket{\phi^+}_{RA}\) and tracing out the classical output \(C\), we obtain the state
\[
\eta^{\mathcal{N}}_{RB}=p\ketbra{\phi^+}{\phi^+}_{RA}+(1-p)\sum_j\ketbra{\psi_j}{\psi_j}_R\otimes\ketbra{\phi_j}{\phi_j}_{A'},
\]
note that here \(\ket{\phi^+}_{RA}\) can be interpreted as an entangled state in
\(\mathcal{H}_{RB}\). We then have 
\[
\mathcal{Q}(\mathrm{D},\ket{\phi^+})\ge I_{\eta^{\mathcal{N}}}(R\rangle B)=p\log_2 d_A,
\]
so that
\[
\mathcal{T}(\mathrm{D})\ge \mathcal{C}(\mathrm{D},\rho_{RA})
+\mathcal{Q}(\mathrm{D},\ket{\phi^+}_{RA})\ge (1-p)\log_2d_{A}+p\log_2d_{A}=\log_2d_{A}.\]
Finally, we complete the proof by saturating the above inequality with the upper-bound on \(\mathcal{T}(\mathrm{D})\) derived in Theorem \ref{t2}.
\end{proof}
Notably, with further investigation, we found no quantum channel other than those in Proposition \ref{p3} and \ref{p4}, for which there exists a Q-DEMUX realization achieving the optimal selector-less demultiplexing strength. 
%
%
%
%
\section{Operational Interpretation of Incompatibility of Quantum Instruments}
A Q-DEMUX at its core configuration involves the instrument realizations of a given quantum channel. Depending upon the nature of the data to be processed, by accessing the selector, the sender Alice can route the quantum system via one of the two instrument realizations of the quantum channel. While in Theorem \ref{t1}, we have characterized the set of quantum channels for which there exists instrument realizations leading to the maximal demultiplexing strength. However, one may tempted to ask about the salient features of these instrument realizations, achieving the optimality. Even, more generally, can the demultiplexing strength \(\mathcal{T}_s(\mathrm{D})\) for any Q-DEMUX \(\mathrm{D}\) could lead to some fundamental characteristics of the instrument realizations? With the help of Lemma \ref{l3} and Theorem \ref{t2} we will now set a threshold on the demultiplexing strength of a Q-DEMUX, violation of which necessarily implies that the two instrument realization corresponding to the induced quantum channel must be incompatible.

The incompatibility of quantum measurements, being a key non-classical feature of quantum systems, leads to various operational signatures, viz., quantum nonlocality \cite{plavala2025all}, quantum steering \cite{cavalcanti2016quantitative}, discrimination of quantum states \cite{skrzypczyk2019all}. Similarly, one could extend the notion of incompatibility for quantum channels \cite{heinosaari2017incompatibility} and finally for quantum instruments as an unification of both the induced POVM measurements and the induced quantum channels \cite{leppajarvi2024incompatibility}. While there various way to bring the notion of incompatibility for quantum instruments \cite{mitra2022compatibility,mitra2023characterizing,leppajarvi2024incompatibility}, here we will consider the traditional one.

\begin{definition}\label{d2}
A pair of quantum instruments \(\mathcal{I}(\Omega,\mathcal{H}_{in},\mathcal{H}_{out})\) and \(\mathcal{J}(\Omega',\mathcal{H}_{in},\mathcal{H}_{out})\) is said to be traditionally compatible, if there exists a joint instrument \(\mathcal{G}(\Omega\times\Omega',\mathcal{H}_{in},\mathcal{H}_{out})\), such that
\begin{align*}
    \sum_{y\in\Omega'}\mathcal{G}_{x,y}=\mathcal{I}_x~\forall x\in\Omega,\text{ and }\sum_{x\in\Omega}\mathcal{G}_{x,y}=\mathcal{J}_y~\forall y\in\Omega'.
\end{align*}
\end{definition}
The definition itself justifies the choice of traditional incompatibility in the present context. In particular, note that
\[\mathcal{I}=\sum_{x\in\Omega}\mathcal{I}_x=\sum_{x\in\Omega}\sum_{y\in\Omega'}\mathcal{G}_{x,y}=\sum_{y\in\Omega'}\sum_{x\in\Omega}\mathcal{G}_{x,y}=\sum_{y\in\Omega'}\mathcal{J}_y=\mathcal{J},\]
that is, the notion of traditional incompatibility deals with those set of quantum instruments for which the induced quantum channel is identical, which is the same constraint we have over Q-DEMUX.  
In other words, any two quantum instruments, with different induced quantum channels, are beyond the scope of traditional compatibility.  

In the following, we will set a bound on the demultiplexing strength \(\mathcal{T}_s(\mathrm{D})\), the violation of which certifies the corresponding pair of instrument realizations must be traditionally incompatible. 
\begin{theorem}\label{t3}
   For any Q-DEMUX \(\mathrm{D}:\mathcal{D}(\mathcal{H}_{in})\to\mathcal{D}(\mathcal{H}_{out})\times X\), if the selector can access two instrument realizations of \(\Lambda_{\text{ind}}\) to achieve \(\log_2 d_{in}<\mathcal{T}_s(\mathrm{D})\leq 2\log_2 d_{in}\), then they must be traditionally incompatible. 
\end{theorem}
\begin{proof}
    While the upper-bound is trivial from Lemma \ref{l3}, we prove the lower-bound by contradiction. 
    
    Suppose, there exists two traditionally compatible instrument realizations \(\mathcal{I}(\Omega,\mathcal{H}_{in},\mathcal{H}_{out})\) and \(\mathcal{J}(\Omega',\mathcal{H}_{in},\mathcal{H}_{out})\) of a quantum channel \(\Lambda_{\text{ind}}:\mathcal{D}(\mathcal{H}_{in})\to\mathcal{D}(\mathcal{H}_{out})\) for a Q-DEMUX \(\mathrm{D}\), such that with these two realizations,  \(\mathcal{T}_s(\mathrm{D})>\log_2 d_{in}\). In particular, the Q-DEMUX \(\mathrm{D}\) implements the instrument \(\mathcal{I}\) for classical information processing as \(\mathrm{D}_0\), i.e., with the selector choice \(s=0\), while the instrument \(\mathcal{J}\) as \(\mathrm{D}_1\) in the context of quantum information processing. 
    Additionally, we denote the correcting operations as \(\{\mathcal{J}^*_y\}_{y\in\Omega}\) for quantum information processing.
    
    By Definition \ref{d2}, we can therefore have a parent instrument \(\mathcal{G}(\Omega\times\Omega',\mathcal{H}_{in},\mathcal{H}_{out})\) such that summing over the random variables \(x\in\Omega\) (\(y\in\Omega'\)) gives \(\mathcal{J}\)(\(\mathcal{I}\)).
As mentioned earlier, \(\sum_{x\in\Omega}\sum_{y\in\Omega'}\mathcal{G}_{x,y}\equiv\Lambda_{\text{ind}}\), which means \(\mathcal{G}\) itself is a valid instrument realization for the induced quantum channel \(\Lambda_{\text{ind}}\). 
    
    Let us then consider another Q-DEMUX \(\mathrm{D}'\) for the induced channel \(\Lambda_{\text{ind}}\) involving the instrument realization \(\mathcal{G}\):\\
    \noindent
    (i) For input quantum state \(\rho\in\mathcal{D}(\mathcal{H}_{in})\), the un-normalized output quantum state at the quantum port \(Q\) is \(\mathcal{G}_z(\rho)\) with a probability \(\Tr[\mathcal{G}_z(\rho)]\), whenever\\
    (ii) \(z:=(x,y)\) is stored in the classical register \(C\).
    
    Moreover, we can associate the same set of correcting operations \(\{\mathcal{J}_y^*\}_y\) for every \(z:=(x,y)\), independent of \(x\in\Omega\). On the other hand, we can identify the output classical random variable as \(x\in X\) for every \(z:=(x,y)\), independent of \(y\in\Omega'\). With this particular postprocessing, it is easy to identify that the classical (quantum) information processing via \(\mathrm{D}'\) effectively implements the instrument realization \(\mathcal{I}\) (\(\mathcal{J}\)), i.e., the Q-DEMUX realization \(\mathrm{D}_0\) (\(\mathrm{D}_1\)).

    Now, the demultiplexing strength (as in Eq. (\ref{e13})) of \(\mathrm{D}'\), with the instrument realization \(\mathcal{G}\), will take the form
    \begin{align*}
        \mathcal{T}(\mathrm{D}')&=\max_{\rho_A\in\mathcal{D}(\mathcal{H}_{in})}[\mathcal{C}(\mathrm{D}',\rho_{RA})+\mathcal{Q}(\mathrm{D}',\ket{\psi}_{RA})]\\
        &\geq \max_{\rho_A\in\mathcal{D}(\mathcal{H}_{in})}[\mathcal{C}(\mathrm{D}_{0},\rho_{RA})+\mathcal{Q}(\mathrm{D}_{1},\ket{\psi}_{RA})]\\
        &=\mathcal{T}_s(\mathrm{D})>\log_2 d_{in}.
    \end{align*}
    Thanks to Theorem \ref{t2}, the selector-less demultiplexing strength is always upper-bounded by \(\log_2d_{in}\), which leads to a contradiction with the last inequality. Therefore, the pair of instrument realizations \(\mathcal{I}\) and \(\mathcal{J}\) for the induced quantum channel \(\Lambda_{\text{ind}}\) must be traditionally incompatible.
\end{proof}
%
%
%
%
\section{Independent Maximization of the demultiplexing strength}
The quantification of the demultiplexing strength, in our analysis, is motivated by the fact that a single quantum state can be used to encode both the classical as well as the quantum information and hence the quantity involves maximization over a single marginal quantum state for both the capacities. Alternatively, one can define independent maximizations for both the quantities. While such an independent maximization lacks its operational meaning in the context of Q-DEMUX, one could interpret the quantity as the joint capacity of quantum channel over its all possible instrument realizations. Interestingly, for perfect Q-DEMUXs, where the sender could access the selector the independent maximization will not change the results. This is because the maximization achieved in Lemma \ref{l3}, remains unaltered even if the two of the capacities are independently maximized. However, beyond the perfect one, the total demultiplexing strength can differ for the independent maximization of the classical and the quantum capacities.

Moreover, in the context of selector-less demultiplexing strength, the situation is completely different even for the optimal one. In fact, there is no immediate non-trivial upper-bound for the selector-less demultiplexing strength, when the classical and the quantum capacity is allowed to maximize independently. As an example, consider the following quantum channel for any arbitrary dimension
\begin{align}\label{e19}
\Lambda_{\text{ind}}^{(d,k)}(\rho)=\sum_{i=0}^{k-1}\bra{i}\rho\ket{i}\ketbra{i}{i}+\sum_{j,l=k}^{d-1}\bra{j}\rho\ket{l}\ketbra{j}{l},~\text{where, }2\leq k<(d-1).\end{align}
Essentially, the channel decoheres a \(k<d\)-dimensional subspace in \(k\)-dimensional computational basis \(\{\ketbra{i}{i}\}_{i=0}^{k-1}\), keeping the rest (\(d-k\))-dimensional subspace invariant. Note that, the channel is similar to that of in Proposition \ref{p4}. In particular, to make it simpler, here we have replaced the \(k\)-dimensional C-Q channel with a classical-classical one and the RI with the identity channel on the \((d-k)\)-dimensional subspace.
The channel then admits an isometric extension of the following form:
\[\mathcal{V}_{\Lambda}^{(d,k)}\ket{\psi}_A=\sum_{i=0}^{k-1}c_i\ket{i}_Q\otimes\ket{i}_E+\sum_{j=k}^{d-1}c_j\ket{j}_Q\otimes\ket{k}_E,~\text{where, }\sum_{i=0}^{d-1}c_i\ket{i}=\ket{\psi}\in\mathcal{H}_{d}.\]
Similar to that of Proposition \ref{p4}, we can now construct a Q-DEMUX \(\mathrm{D}\), for \(\Lambda_{\text{ind}}\), involving an instrument realization \(\mathcal{I}(\Omega:=\{0,\cdots, k\}, \mathcal{H}_{in},\mathcal{H}_{in})\) corresponding to the above isometry. In particular, consider 
\[\mathcal{I}_x=\langle x|\cdot|x\rangle \ketbra{x}{x},~x\in\{0,\cdots,k-1\}~\&~\mathcal{I}_{k}=\sum_{j,l=k}^{d-1}\langle j|\cdot |l\rangle \ketbra{j}{l}.\]
Accordingly, the outcome of the Q-DEMUX \(\mathrm{D}\) for \(\Lambda_{\text{ind}}^{(d,k)}\), using the above instrument realization becomes
\begin{align}\label{e20}
    \mathrm{D}(\rho_A)=\sum_{i=0}^{k-1}\langle i|\rho|i\rangle\ketbra{i}{i}_Q\otimes\ketbra{i}{i}_C+\left(\sum_{j,l=k}^{d-1}\ketbra{j}{j}\rho\ketbra{l}{l}\right)_Q\otimes\ketbra{k}{k}_C,~\forall\rho_A\in\mathcal{D}(\mathcal{H}_{in}).
\end{align}
We complete the configuration of \(\mathrm{D}\) by assigning the correcting operation \(id_d\) (no correction), independent of the indices registered in the classical port \(C\).

Now, suppose an input classical random variable \(y\in Y\) with \(|Y|=k+1\), is encoded in the quantum states \(\{\{\ketbra{i}{i}\}_{i=0}^{k-1},~\tilde{\rho}\}\), where \(\tilde{\rho}\in\text{Span}\{\ket{k},\cdots,\ket{d-1}\}\). Then the classical quantum state between the reference system and Alice takes the form:
\[\tilde{\rho}_{RA}(p_y)=\sum_{y=0}^{k-1}p_y\ketbra{y}{y}_R\otimes\ketbra{y}{y}_A+p_k\ketbra{k}{k}_R\otimes\tilde{\rho}_A,\text{ with }\rho_A=\Tr_R[\tilde{\rho}_{RA}]=\sum_{y=0}^{k-1}p_y\ketbra{y}{y}+p_k\tilde{\rho}.\]
Since the outcome of the projective measurement \(\{\ketbra{x}{x}\}_{x=0}^k\) on the subsystem \(E\) is registered in the classical output port \(C\) of the Q-DEMUX, it is now trivial to see that \(p(x|y)=\delta_{x,y}\) for every \(x,y\in\{0,\cdots,k\}\). Therefore the classical capacity of the Q-DEMUX, with the classical quantum state \(\tilde{\rho}_{RA}\) is
\[\mathcal{C}(\mathrm{D},\tilde{\rho}_{RA}(p_y))=\max_{\{p_y\}_{y=0}^k}I(Y:X)=\log_2(k+1).\]
 On the other hand, with this same isometric realization, consider the transmission of subsystem \(A\) of the (\(d-k\))-dimensional maximally entangled state 
 \[\ket{\phi^+_{d-k}}_{RA}=\frac1{\sqrt{d-k}}\sum_{j=k}^{d-1}\ket{j}_{R}\otimes\ket{j}_A,\text{ where }\rho'_A=\Tr_R[\ketbra{\phi^+_{d-k}}{\phi^+_{d-k}}]=\frac1{d-k}\sum_{j=k}^{d-1}\ketbra{j}{j}.\]
 Using Eq.(\ref{e20}), for \(\rho'_A\) as the input to the Q-DEMUX only \(x=k\) will be registered in the classical register \(C\) and the output quantum state will remain invariant. Therefore, the quantum capacity of the Q-DEMUX with this specific choice will become,
 \[\mathcal{Q}(\mathrm{D},\ket{\phi^+_{d-k}}_{RA})=I_{\phi^+_{d-k}}(R\rangle Q)=\log_2(d-k).\]
 Evidently, the selector-less total demultiplexing strength of the Q-DEMUX \(\mathrm{D}\), under independent maximization, becomes
 \begin{align*}
 \mathcal{T}^{\text{IM}}(\mathrm{D})&=
 \max_{\rho_A\in\mathcal{D}(\mathcal{H}_{in})}\mathcal{C}(\mathrm{D},\rho_{RA})+\max_{\rho'_A\in\mathcal{D}(\mathcal{H}_{in})}\mathcal{Q}(\mathrm{D},\ket{\psi}_{RA})\\&\geq\mathcal{C}(\mathrm{D},\tilde{\rho}_{RA}(p_y))+\mathcal{Q}(\mathrm{D},\ket{\phi^+_{d-k}}_{RA})=\log_2(k+1)(d-k)
 \end{align*}
 Finally, for every \(d>3\) and \(2\leq k<(d-1)\) it is easy to check that \(\log_2(k+1)(d-k)>\log_2d\). 
 
 It is important to mention that in such a scenario, by looking at the classical register of \(\mathrm{D}\), it is possible to predict the nature of the information Alice is trying to convey. In particular, if the classical register \(C\) outputs \(x\in\{0,\cdots,k-1\}\) then the information must be classical, specifically \(y=x\); While for \(x=k\), it can both quantum as well as classical. Moreover, after receiving the encoded quantum state, if Alice perform a projective measurement \(\{\{\ketbra{i}{i}\}_{i=0}^{k-1},\sum_{j=k}^{d-1}\ketbra{j}{j}\}\) on it, she could predict the information is classical whenever any other than the last projector clicks; In fact, she could know the exact classical message at all of these instances. Therefore, the information is no more oblivious to her.

 Interestingly, the channel in Eq. (\ref{e19}) can also be written in form of a random unitary channel, given by
 \begin{align}\label{e21}
     \Lambda_{\text{ind}}^{(d,k)}(\rho)=\frac1k\sum_{m=0}^{k-1}V_m\rho V_m^{\dagger},\text{ where, }\mathcal{L}(\mathcal{H}_{d})\ni V_m\equiv U_m\oplus I_{d-k}~\&~ U_m=\sum_{j=0}^{k-1}e^{i\frac{2\pi(m.j)}k}\ketbra{j}{j},~\forall m\in\{0,\cdots,k-1\},
 \end{align}
 where, \(I_{d-k}\) is the identity operator on the \((d-k)\)-dimensional subspace of \(\mathcal{H}_d\), spanned by \(\{\ket{k},\cdots,\ket{d-1}\}\). Therefore, \(\Lambda_{\text{ind}}^{(d,k)}\) also admits the following isometric extension
 \[\tilde{\mathcal{V}}_{\Lambda}^{(d,k)}\ket{\psi}_A=\frac1{\sqrt{k}}\sum_{m=0}^{k-1}V_m\ket{\psi}_Q\otimes\ket{m}_E,~\text{for every }\ket{\psi}\in\mathcal{H}_d.\]
 Now, if Alice has access to the selector \(s\), then during the quantum information processing she could use another Q-DEMUX realization \(\tilde{\mathrm{D}}\) for the same \(\Lambda_{\text{ind}}\), invoking another instrument realization using the above isometry. Then the action of \(\tilde{\mathrm{D}}\) can be written as,
    \[\tilde{\mathrm{D}}(\rho_A)=\frac1k\sum_{m=0}^{k-1}\left(V_m\rho V_m^{\dagger}\right)_Q\otimes\ketbra{m}{m}_C.\]
    We can now define a complete Q-DEMUX \(\mathrm{D}_s\), with the access to the selector, for the quantum channel \(\Lambda_{\text{ind}}^{(d,k)}\): For the classical information processing with \(s=0\), it is \(\mathrm{D}_0\equiv \mathrm{D}\), while for the quantum information processing with \(s=1,~\mathrm{D}_1\equiv\tilde{\mathrm{D}}\). In both the cases, we can assign the classical output as \(x\in X:=\{0,\cdots,k\}\) in the classical register \(C\), along with the correcting operations \(\{V_x^{\dagger}\}_{x=0}^{k-1}\) and \(id_Q\) (no correction) for \(x=k\).

    Now, for the classical information processing, with \(s=0\), along with the same \(\tilde{\rho}_{RA}(p_y=\frac1{k+1})\) exactly \(\log_2(k+1)\)-bits can be communicated perfectly. On the other hand, for quantum information with \(s=1\) a \(d\)-dimensional maximally entangled state \(\ket{\phi^+_d}\) can be established between reference system and quantum output port \(Q\), using the random unitary realization. This implies, that by accessing the selector, along with independent maximization for both classical and quantum capacity, for a Q-DEMUX \(\mathrm{D}\) with induced quantum channel \(\Lambda_{\text{ind}}^{(d,k)}\),
    \[\mathcal{T}^{\text{IM}}_s(\mathrm{D})\geq \log_2d(k+1).\]
    Note that, for every \(k<d-1\) the quantity \(\log_2d(k+1)<2\log_2d\), respecting the bound derived in Lemma \ref{l3}.
%
%
%
%
\section{Discussions}
In summary, we have extended the notion of classical demultiplexer in the quantum regime, where both the classical as well as quantum information can be processed through the same transmission line. Construction of such a generalized information demultiplexer involves two of the different instrument realization of the same quantum channel, accessing which the sender is allowed to demultiplex the encoded information in the respective port. For any quantum channel, which admits both random isometry as well as classical-quantum representation, by accessing these two instrument realizations Alice can respectively demultiplex the quantum and classical information, encoded in the same density matrix, perfectly. While the implications of quantum instruments is well-known in quantum computing models \cite{buzek2006programmable, ji2024incompatibility, sau2025sequential}, our work highlights a new operational insights for the same in the context of generalized information processing. Additionally, our results provides an operational interpretation of incompatible quantum instruments in terms of their demultiplexing strength. This connection suggests a fundamental trade-off between the classical and quantum information transmission abilities of a quantum instrument, allowing all possible classical post-processing on it. In particular, any quantum instrument classically post-processed from another one must be traditionally compatible in nature  \cite{leppajarvi2021postprocessing,mitra2022compatibility}. Therefore, any quantum instrument, along with classical post-processing, can not transfer both the classical and the quantum information perfectly. In fact, if a single quantum state is used to encode both the classical and the quantum information (depending upon the correlation shared with the reference system) then the total information processing strength can not go beyond \(\log_2d_{in}\)-bits. However, if the given instrument allows independent ensembles to maximize their total information processing ability this bound can be violated for any arbitrary dimension.  

It is also important to mention that our results can also be interpreted as a highly restricted settings of information exchange in an environment assisted communication model. Notably, the assistance of a friendly environment in context of transferring quantum resources has been exemplified in several theoretical \cite{gregoratti2004quantum, winter2005environment, hayden2005correcting, karumanchi2016classical, chowdhury2025minimal, sarkar2025lost} and experimental premises \cite{pirandola2021environment}. Since, every quantum instrument can be seen as a specific measurement performed on the environment system, the quantum information processing via the Q-DEMUX can be explained in terms of the environment assisted quantum communication. However, the present scenario is restrictive in the sense that a Q-DEMUX by construction only allow to access the environment in each of the runs individually; no joint access on the environment of the multiple quantum channels are permissible. On the other hand, for the classical information processing, since the classical information decoded at the environment's end, it can be interpreted as the single-shot classical capacity of the complementary channel, without any further assistance from the main channel. Therefore, the bound on the selector-less demultiplexing strength, in other words, identify the maximal amount of generalized information can be processed through a quantum channel, with an one-sided measurement assistance from the environment.

Finally, our result paves the way for various open directions for both in the contexts of quantum circuit architecture and general information theory. While the proposed Q-DEMUX is solely intended to segregating the generalized information encoded in a quantum state with a binary selector, in a practically motivated scenario one could expect to demultiplex the respective information further among multiple quantum or classical ports. A simple way to configure the same is just to implement a pair of classical and quantum demultiplexer in the respective ports of the Q-DEMUX, but that may not be the most efficient circuit realization in terms of the consumed gates and the ancillary systems. This naturally instigates about the experimental manifestations of our findings in table-top set-up, further with successful implementation in the quantum networks. On the other hand, although our results set a bound on the demultiplexing strength of compatible quantum instruments, it remains unaddressed whether the demultiplexing strength could be a valid measure for the traditional incompatibility of quantum instruments. Moreover, the instances of independent maximization of classical and quantum information processing ability can be seen as a first step towards characterizing the true information processing ability of a quantum instrument. The further characterization of which could be an important direction for operational implications of the quantum instruments in general. 

\medskip
\textit{Acknowledgments.} This work is supported by the Slovak Research and Development Agency through Grant no. APVV-22-0570, the Scientific Grant Agency of the Ministry of Education, Slovak Republic through Grant no. VEGA 2/0128/24. TG acknowledges the Štefan Schwarz Support Fund 2025/OV1/046 by the Slovak Academy of Sciences.

\bibliography{reference}

\end{document}